\documentclass[a4paper]{article}%
\pdfoutput=1

 \newcommand{\sv}[1]{}%
 \newcommand{\lv}[1]{#1}%
\usepackage{etoolbox}
\newcommand{\appendixText}{}

 \newcommand{\toappendix}[1]{#1}%

\usepackage[hypertexnames=false]{hyperref} %
\sv{}
\usepackage[svgnames]{xcolor} %
\lv{\usepackage{lmodern}}
\usepackage{a4wide} %

\usepackage[utf8]{inputenc}
\lv{\usepackage{color}
\usepackage[protrusion=true,expansion=true]{microtype}
\usepackage[T1]{fontenc}
}
\usepackage{amsfonts}
\usepackage{amsopn, amsmath}
\usepackage{amsthm}
\usepackage[mathscr]{euscript} %
\usepackage{bm}
\usepackage{subcaption}
\usepackage{amssymb,amscd,boxedminipage}
\usepackage{graphicx}
\graphicspath{{./}{../}}
\usepackage{xspace}	
\usepackage{paralist}
\lv{\usepackage{authblk}}
\usepackage[linesnumbered, vlined]{algorithm2e}
\usepackage[absolute]{textpos}
\usepackage{tikz}
\usepackage{verbatim, float} %
\usepackage{placeins}
\usetikzlibrary{calc,positioning}
\usepgflibrary{shapes,arrows}
\usepgflibrary{decorations.pathreplacing}
\usepackage{enumerate}%

\lv{\usepackage{marginfix} %
\setlength{\marginparwidth}{2.5cm}}
\usepackage[color=Olive,textwidth=2.5cm]{todonotes} 
\presetkeys{todonotes}{fancyline}{}

\usepackage{tabularx}

\newcommand{\NP}{\ensuremath{\mathsf{NP}}\xspace}

\newcommand{\df}{:=}

\newcommand{\Oh}[1]{\ensuremath{\mathcal{O}(#1)}}

\def\B{\ensuremath{{\mathcal B}}\xspace}
\def\U{\ensuremath{{\mathcal U}}\xspace}

\lv{
\title{
 \U-Bubble Model for Mixed Unit Interval Graphs and its Applications:\\
The MaxCut Problem Revisited%
  \thanks{Jan Kratochvíl was supported by grant GAČR 19-17314J of the Czech National Science Foundation.
    Tomáš Masařík and Jana Novotná received funding from the European
    Research Council (ERC) under the European Union's Horizon 2020
    research and innovation programme Grant Agreement 714704, and
    from Charles University student grant SVV-2017-260452.
  } \thanks{An extended abstract of this manuscript will appear at Mathematical Foundations of Computer Science (MFCS) 2020~\cite{U-BubbleMFCS}}
}

\author[1]{Jan Kratochvíl}
\author[1,2]{Tomáš Masařík}
\author[1,2]{Jana Novotná}

\affil[1]{Charles University, Prague, Czech Republic}
\affil[2]{University of Warsaw, Poland}
\affil[ ]{\texttt{\{honza,masarik,janca\}@kam.mff.cuni.cz}}
\date{}
}

\def\povolenyobrazky{1}

\newcommand{\R}{\mathbb{R}}
\newcommand{\N}{\mathbb{N}}

\definecolor{darkblue}{rgb}{0,0,0.45}
\definecolor{darkred}{rgb}{0.6,0,0}
\definecolor{darkgreen}{rgb}{0.13,0.5,0}
\hypersetup{colorlinks, linkcolor=darkblue, citecolor=darkgreen,
urlcolor=darkblue}

\newcommand{\drawBubble}[6]{

\draw[line width=0.065cm] (#1,#2) rectangle (#1+2,#2+2);
\draw (#1,#2+1) -- (#1+2,#2+1);
\draw (#1+1,#2) -- (#1+1,#2+2);
\node (a) at (#1+0.5,#2+1.5) [align=left] {\large $#3$}; %
\node (b) at (#1+0.5,#2+0.5) {$#4$}; %
\node (c) at (#1+1.5,#2+1.5) {$#5$}; %
\node (d) at (#1+1.5,#2+0.5) {$#6$}; %

    }

\newcommand{\openclosed}{\ensuremath{{\mathcal U}^{\pm}}\xspace} 
\newcommand{\semimixed}{\ensuremath{{\mathcal U}^{\pm,+-}}\xspace}
\newcommand{\mixed}{\ensuremath{{\mathcal U}}\xspace}
\newcommand{\I}{\ensuremath{{\mathcal I}}\xspace}

\newcommand{\Mis}[1]{\ensuremath{{\alpha(#1)}}\xspace}
\newcommand{\Mclique}[1]{\ensuremath{{\omega(#1)}}\xspace}
\newcommand{\GB}{\ensuremath{G({\mathcal B})}\xspace}
\newcommand{\cutset}{\ensuremath{E(S,\overline S)}\xspace}

\newcommand{\Bij}{\ensuremath{B_{i,j}}\xspace}

\newcommand{\btop}[1]{\ensuremath{\operatorname{top}(#1)}\xspace} %
\newcommand{\col}[1]{\ensuremath{\operatorname{col}(#1)}\xspace} %
\newcommand{\row}[1]{\ensuremath{\operatorname{row}(#1)}\xspace} %
\newcommand{\typeI}[2]{\ensuremath{\operatorname{type}_#1(#2)}\xspace} %
\newcommand{\type}[1]{\ensuremath{\operatorname{type}(#1)}\xspace} %
\newcommand{\BT}[1]{\ensuremath{{B_{\btop{#1}, #1}}}\xspace} %
\newcommand{\Boo}[1]{\ensuremath{B_{#1}^{--}}\xspace}
\newcommand{\Boc}[1]{\ensuremath{B_{#1}^{-+}}\xspace}
\newcommand{\Bco}[1]{\ensuremath{B_{#1}^{+-}}\xspace}
\newcommand{\Bcc}[1]{\ensuremath{B_{#1}^{++}}\xspace}
\newcommand{\Bhc}[1]{\ensuremath{B_{#1}^{*+}}\xspace}
\newcommand{\Bch}[1]{\ensuremath{B_{#1}^{+*}}\xspace}

\newcommand{\boo}[1]{\ensuremath{b_{#1}^{--}}\xspace}
\newcommand{\boc}[1]{\ensuremath{b_{#1}^{-+}}\xspace}
\newcommand{\bco}[1]{\ensuremath{b_{#1}^{+-}}\xspace}
\newcommand{\bcc}[1]{\ensuremath{b_{#1}^{++}}\xspace}

\newcommand{\maxcut}{MaxCut\xspace} %
\newcommand{\mcs}[1]{\ensuremath{mcs(#1)}} %
\newcommand{\cwd}[1]{\ensuremath{cwd(#1)}} %
\newcommand{\relabel}[2]{\ensuremath{\rho_{#1\rightarrow #2}}} %
\newcommand{\edge}[2]{\ensuremath{\eta_{#1,#2}}} %

\newcommand{\li}[1]{\ensuremath{\ell(#1)}\xspace}
\newcommand{\ri}[1]{\ensuremath{{r}(#1)}\xspace}

\newcommand{\ol}{\overline}

\newcommand{\Eng}[1]{}

\sv{
\newtheorem{thm}{Theorem}
\newtheorem{observation}[thm]{Observation}
\newtheorem{cor}[thm]{Corollary}
}
\lv{
\theoremstyle{plain}
\newtheorem{thm}{Theorem}
\newtheorem{lemma}[thm]{Lemma}

\newtheorem{observation}[thm]{Observation}
\newtheorem{cor}[thm]{Corollary}

\theoremstyle{plain}

\newtheorem{definition}{Definition}

\theoremstyle{remark}

\newtheorem*{example}{Example}
}

\lv{\overfullrule=1mm}

\begin{document}

\maketitle
\lv{
\begin{textblock}{20}(0, 12)
\includegraphics[width=40px]{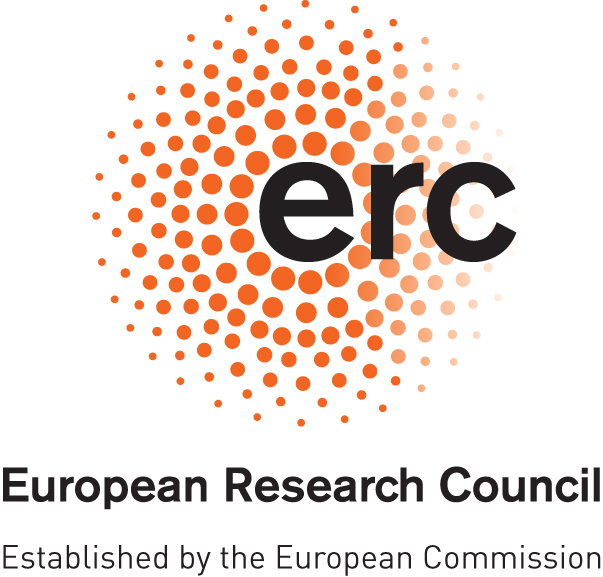}%
\end{textblock}
\begin{textblock}{20}(-0.25, 12.9)
\includegraphics[width=60px]{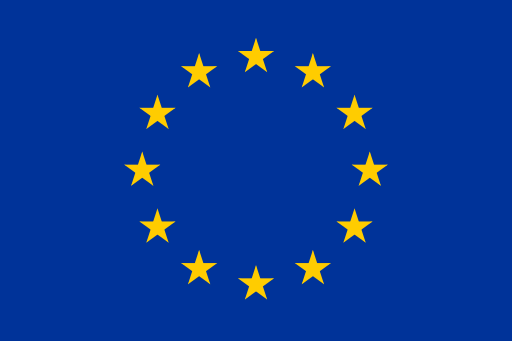}%
\end{textblock}
}

\begin{abstract}
Interval graphs, intersection graphs of segments on a real line (intervals), play a key role in the study of algorithms and special structural properties. 
Unit interval graphs, their proper subclass, where each interval has a unit length, has also been extensively studied.
We study mixed unit interval graphs---a generalization of unit interval graphs where each interval has still a unit length, but %
intervals of more than one type (open, closed, semi-closed) are allowed. 
This small modification 
captures a much richer class of graphs.
In particular, mixed unit interval graphs are not claw-free, compared to unit interval graphs. 

Heggernes, Meister, and Papadopoulos 
defined a representation of unit interval graphs called 
the 
bubble model which turned out to be useful in algorithm design.
We extend this model to the class of mixed unit interval graphs and demonstrate the advantages of this generalized model by providing a subexponential-time algorithm for solving the MaxCut problem on mixed unit interval graphs. 
In addition, we derive a polynomial-time algorithm for certain subclasses of mixed unit interval graphs. 
We point out a substantial mistake in the proof of the polynomiality of the MaxCut problem on unit interval graphs by  Boyaci, Ekim, and Shalom (2017). Hence, the time complexity of this problem on unit interval graphs remains open.
We further provide a better algorithmic upper-bound on the clique-width of mixed unit interval graphs.
\lv{Clique-width is one of the most general structural graph parameters, where a large group of natural problems is still solvable in the tractable time when an efficient representation is given.
Unfortunately, the exact computation of the clique-width representation is \NP-hard. Therefore, good upper-bounds on clique-width are highly appreciated, in particular, when such a bound is algorithmic.}

\end{abstract}

\section{Introduction}

A graph $G$ is an \emph{intersection graph} if there exists a family of nonempty sets $\mathcal{F}=\{S_1,\dots,S_n\}$ such that for each vertex $v_i$ in $G$, a set $S_i\in\mathcal{F}$ is assigned in a way that there is an edge $v_iv_j$ in $G$ if and only if $S_i\cap S_j\neq\emptyset$.
We say that $G$ has an \emph{$\mathcal{F}$-intersection representation}.
Any graph can be represented as an intersection graph since per each vertex, we can use the set of its incident edges. However, many important graph classes can be described as intersection graphs with a restricted family of sets. Depending on the geometrical representation, different types of intersection graphs are defined, for instance, interval, circular-arc, disk graphs, etc.
\\\indent
\emph{Interval graphs} are intersection graphs of segments of the real line, called intervals.
Such a representation is being referred to as \emph{interval representation}.
They have been a well known and widely studied class of graphs from both the theoretical and the algorithmic points of view\lv{ since 1957}.
\lv{They were first mentioned independently in combinatorics (Hajos, 1957~\cite{Hajos57, Brandstadt99}) and genetics (Benzer, 1959~\cite{Benzer59}).}%
\lv{%

}%
Interval graphs have a nice structure, they are chordal and, therefore, also perfect which provides a variety of graph decompositions and models. %
Such properties are often useful tools for the algorithm design---the most common algorithms on them are based on dynamic programming.  
Therefore, many classical \NP-hard problems are polynomial-time solvable on interval graphs,
for instance Hamiltonian cycle \lv{(Keil 1985~\cite{Keil85})}\sv{\cite{Keil85}}, Graph isomorphism \lv{(Booth, 1976~\cite{BoothL76})}\sv{\cite{BoothL76}} or  Colorability \lv{(Golumbic, 1980 \cite{Golumbic80})}\sv{\cite{Golumbic80}}  are solvable even in linear time.  
Surprisingly, the complexity of some well-studied problems is still unknown despite extensive research, e.g.\ the $L_{2,1}$-labeling problem, or the packing coloring problem. 
\lv{%
Interval graphs have many real applications in diverse areas including genetics~\cite{Benzer59}, economics, and archaeology~\cite{Roberts00, Roberts78Apl}. 
According to Golumbic~\cite{Golumbic80}, many real-world applications involve solving problems on graphs which are either interval graphs themselves or are related to interval graphs in a natural way.}%
\\\indent
An important subclass of interval graphs is the class of \emph{proper interval graphs}, graphs which can be represented by such an interval representation that no interval properly contains another one. Another interval representation is a representation with intervals (of the same type) of only unit lengths, graphs which have such a representation are called \emph{unit interval graphs}.
Roberts proved\lv{ in 1969}~\cite{Roberts69} that a graph is a proper interval graph if and only if it is a unit interval graph. 
Later, Gardi came up with a constructive combinatorial proof~\cite{Gardi07}.
\lv{%

}%
The mentioned results do not specifically care about what types of intervals (open, closed, semi-closed) are used in the interval representation. 
However, as far as there are no restrictions on lengths of intervals, it does not matter which types of intervals are used~\cite{Shuchat14}.
The same applies if there is only one type of interval in the interval representation.
However, this is not true when all intervals in the interval representation have unit length and at least two types of intervals are used.
In particular, the claw $K_{1,3}$ can be represented using one open interval and three closed intervals.
\\\indent
Recently, it has been observed that a restriction on different types of intervals in the unit interval representation leads to several new subclasses of interval graphs. 
We denote the set of all open, closed, open-closed, and closed-open intervals of \emph{unit length} by $\U^{--}$, $\U^{++},$  $\U^{-+}$, and $\U^{+-}$, respectively.
Let $\U$ be the set of all types of unit intervals.
Although there are 16 different combinations of types of unit intervals, it was shown in \cite{Dourado12,RautenbachS13,Shuchat14,Joos15,TalonK18} in the years 2012--2018 that they form only four different classes of mixed unit interval graphs. 
In particular, the following closure holds:
\lv{\begin{align*}
\emptyset \subsetneq 
\text{unit}&\text{ interval}
\subsetneq \text{unit open and closed interval} %
\subsetneq \text{semi-mixed unit interval} \subsetneq \\
           &\text{mixed unit interval} 
\subsetneq \text{interval graphs,} 
\end{align*}
}%
\sv{%
  $
\emptyset \subsetneq 
\text{unit}\text{ interval}
\subsetneq \text{unit open and closed interval} %
\subsetneq 
$
$
\text{semi-mixed unit interval} \subsetneq %
           \text{mixed unit interval} 
\subsetneq \text{interval graphs} 
$,
}%
where \emph{unit open and closed interval} graphs have $\left(\U^{++}\cup\U^{--}\right)$-representation, \emph{semi-mixed unit interval} graphs have $\left(\U^{++}\cup\U^{--}\cup\U^{-+}\right)$-representation, and \emph{mixed unit interval} graphs have $\U$-representation.
Hence, mixed unit interval graphs allow all types of intervals of unit length.

\begin{definition} 
A graph $G$ is a \emph{mixed unit interval graph} if it has a $\U$-intersection representation. 
We call such representation a \emph{mixed unit interval representation.}
\end{definition}

 There are lots of characterizations of interval and unit interval graphs.
Among many of the characterizations, we single out a matrix-like representation called \emph{bubble model}~\cite{HeggernesMP09}.
A similar notion was independently discovered by Lozin~\cite{Lozin2011} under the name canonical partition.
In the bubble model, vertices of a unit interval graph $G$ are placed into a ``matrix'' where each matrix entry may contain more vertices as well as it can be empty. Edges of $G$ are represented implicitly with quite strong conditions: each column forms a clique; and in addition, edges are only between consecutive columns where they form nested neighborhood (two vertices $u$ and $v$ from consecutive columns are adjacent if and only if $v$ occurs in a higher row than $u$). 
In particular, there are no edges between non-consecutive columns.
This representation can be computed and stored in linear space given a proper interval ordering representation.
\lv{%

}%
We introduce a similar representation of mixed unit interval graphs, called \emph{\U-bubble model}, and we extend some results from unit interval graphs to mixed unit interval graphs using this representation.
 The representation has almost the same structure as the original bubble model, except that edges are allowed in the same row under specific conditions. 
 We show that a graph is a mixed unit interval graph if and only if it can be represented by a \U-bubble model. 

\begin{thm} \label{thm_bubble}
A graph is a mixed unit interval graph if and only if it has a \U-bubble model. 
Moreover, given a mixed unit interval representation of graph $G$ on $n$ vertices, a \U-bubble model can be constructed in $\Oh{n}$ time.
\end{thm}

\lv{In addition, we show properties of our model, such as the relation of the size of a maximum independent set or maximum clique, and the size of the model, see Subsection~\ref{s_bubble_prop}.}

Given a graph $G$, the \maxcut problem is a problem of finding a partition of vertices of $G$ into two sets $S$ and $\ol S$ such that the number of edges with one endpoint in $S$ and the other one in $\ol S$ is maximum among all partitions.
There were two results about polynomiality of the \maxcut problem in unit interval graphs in the past years; the first one by Bodlaender, Kloks, and Niedermeier in 1999~\cite{BodlaenderKN99},
 the second one by Boyaci, Ekim, and Shalom which has been published in 2017~\cite{BoyaciES17}. 
 The result of the first paper was disproved by authors themselves a few years later~\cite{BodlaenderDisproof}. 
In the second paper, the authors used a bubble model for proving the polynomiality. 
However, we realized that this algorithm is also incorrect.
Moreover, it seems to us to be hardly repairable. 
We provide further discussion and also a concrete example, in Subsection~\ref{s_example}.
The complexity of the \maxcut problem in interval graphs was surprisingly unknown for a long time.
Interestingly, a result about \NP-completeness by Adhikary, Bose, Mukherjee, and Roy has appeared on arXiv~\cite{maxcut-hard} very recently\footnote{After the submission of the conference version of this paper.}.

Using the \U-bubble model, we obtain at least a subexponential-time algorithm for \maxcut in mixed unit interval graphs. 
We are not aware of any subexponential algorithms on interval graphs.
In general graphs, there has been extensive research dedicated to approximation of \maxcut in subexponential time, see e.g.~\cite{Arora2015} or~\cite{Hopkins2019}.
Furthermore, we obtain a polynomial-time algorithm if the given graph has a \U-bubble model with a constant number of columns. This extends a result by Boyaci, Ekim, and Shalom~\cite{BoyaciES18} who showed a polynomial-time algorithm for \maxcut on unit interval graphs which have a bubble model with two columns (also called co-bipartite chain graphs).
The question of whether the \maxcut problem is polynomial-time solvable or \NP-hard in unit interval graphs still remains open.

\begin{thm}\label{thm_maxcut}
Let $G$ be a mixed unit interval graph.
The maximum cardinality cut can be found in time $2^{\tilde{0}(\sqrt{n})}.$
\end{thm}

\begin{cor}\label{c_columns}
The size of a maximum cut in the graph class defined by \U-bubble models with $k$ columns can be determined in the time ${\cal O}(n^{k+5})$. Moreover, for $k=2$ in time ${\cal O}(n^5)$.
\end{cor}

\lv{he third part of the paper is devoted to clique-width, one of the graph parameters that is used to measure the complexity of a graph.}
Many \NP-hard problems can be solved efficiently on graphs with bounded clique-width~\cite{CourcelleMR00}.
In general, 
\lv{it is \NP-complete to compute the exact value of clique-width. Furthermore, }%
it is \NP-complete even to decide if the graph has clique-width at most $k$ for a given number $k$, see \cite{Fellows09}.

Unit interval graphs are known to have unbounded clique-width~\cite{Golumbic00}.
It follows from results by Fellows, Rosamond, Rotics, and Szeider~\cite{FellowsRRS06},\ and Kaplan and Shamir~\cite{KaplanS96} that the clique-width of (mixed) unit interval graphs is upper-bounded by 
$\omega$ (the maximum size of their clique)~$+ 1$. 
Heggernes, Meister, and Papadopoulos~\cite{HeggernesMP09} improved this result for unit interval graphs using the bubble model.
There, the clique-width is upper-bounded by a minimum of $\alpha$ (the maximum size of an independent set) + 1, 
and a parameter related to the bubble model representation which is in the worst case $\omega + 1$.
We use similar ideas to extend these bounds to mixed unit interval graphs using the \U-bubble model. 
In particular, we obtain that the upper-bound on clique-width is the minimum of the analogously defined parameter for a \U-bubble model and 
$2\alpha+3$. 
The upper-bound is still in the worst case 
$\omega+ 1$.
The upper-bound can be also expressed in the number of rows or columns of \U-bubble model.
Refer to Theorem~\ref{thm_main_cwd} and Corollary~\ref{cor_cwd_rows_columns} in Section~\ref{s_cwd} for further details.
As a consequence, we obtain an analogous result to Corollary~\ref{c_columns} for rows using the following result.
Fomin, Golovach, Lokshtanov, and Saurabh~\cite{Fomin2014} showed that the \maxcut problem can be solved in time $\Oh{n^{2t+\Oh{1}}}$ where $t$ is clique-width of the input graph. %
By the combination of their result and our upper-bounds on clique-width (Theorem~\ref{thm_main_cwd} in Section~\ref{s_cwd}) we derive not only polynomial-time algorithm when the number of columns is bounded (with worse running time) but also a polynomial-time algorithm when the number of rows is bounded, formulated as Corollary~\ref{c_rows}.

\begin{cor}\label{c_rows}
  The size of a maximum cut in the graph class defined by \U-bubble models with $k$ rows can be determined in the time ${\cal O}(n^{4k+\Oh{1}})$.
\end{cor}

\lv{\subsection{Preliminaries and Notation}}
\sv{\noindent{\bf Preliminaries and Notation.~~}}
By a \emph{graph} we mean a finite, undirected graph without loops and multiedges. Let $G$ be a graph. We denote by $V(G)$ and $E(G)$ the vertex and edge set of $G$, respectively; with $n=|V(G)|$ and $m=|E(G)|$.
Let $\Mis{G}$ and $\Mclique{G}$ denote the maximum size of an independent set of $G$ and the maximum size of a clique in $G$, respectively.
By a \emph{family} we mean a multiset $\{S_1,\dots,S_n\}$ which allows the possibility that $S_i=S_j$ even though $i\neq j$.

Let $x, y \in \R$ be real numbers. We call 
the set $\{z\in\R:\,x\le z\le y\}$ \emph{closed interval} $[x,y]$, 
the set $\{z\in\R:\,x<z<y\}$ \emph{open interval} $(x,y)$, 
the set $\{z\in\R:\,x<z\le y\}$ \emph{open-closed interval}  $(x,y]$, and  
the set $\{z\in\R:\,x\le z<y\}$ \emph{closed-open interval} $[x,y)$.
By \emph{semi-closed interval} we mean interval which is open-closed or closed-open.
We denote the set of all open, closed, open-closed, and closed-open intervals of unit length by $\U^{--}$, $\U^{++},$  $\U^{-+}$, and $\U^{+-}$, respectively.
Formally,
$\U^{++} \df \{[x,x+1]:\, x\in \R\},$ 
$\U^{--} \df \{(x,x+1):\, x\in \R\},$ 
$\U^{+-} \df \{[x,x+1):\, x\in \R\},$
and $\U^{-+} \df \{(x,x+1]:\, x\in \R\}.$ 
We further denote the set of all unit intervals
by 
\lv{\[}
\sv{$}
\U \df\U^{++}\cup\U^{--}\cup\U^{+-}\cup\U^{-+}.
\sv{$}
\lv{\]}
From now on, we will be speaking only about unit intervals.
\lv{%

}%
Let $I$ be an interval, we define the left and right end of $I$ as $\li{I} \df \inf(I)$ and $\ri{I} \df \sup(I)$, respectively. 
Let $I,J\in \U$ be unit intervals, $I,J$ are \emph{almost twins} if $\li{I}=\li{J}.$ 
The \emph{type of an interval} $I$ is a pair $(r,s)$ where $I\in\U^{r,s},\,r,s\in\{+,-\}.$ 

Let $G=(V,E)$ be a graph and $\mathcal{I}$ an interval representation of $G$. Let $v\in V$ be represented by an interval $I_v\in\U^{r,s}$, where $r,s\in\{+,-\}$, in $\mathcal{I}$. The \emph{type of a vertex} $v\in V$ in $\I$, denoted by $\typeI{\I}{v}$, is the pair $(r,s)$. We use \type{v} if it is clear which interval representation we have in mind.
We follow the standard approach where the maximum over the empty set is~$-\infty$.
The notion of $\tilde{\cal O}$ denotes the standard ``big 0'' notion which ignores polylogarithmic factors,~i.e, ${\cal O}(f(n) \log^k n) = \tilde{{\cal O}}(f(n))$, where $k$ is a constant.

\toappendix{
  \sv{\section{Recognition and \U-representation of mixed unit interval graphs}\label{s_recog}}
\lv{\subsubsection{Recognition and \U-representation of mixed unit interval graphs}}

All the classes of mixed unit interval graphs can be characterized using forbidden induced subgraphs, sometimes by infinitely many.
Rautenbach and Szwarcfiter~\cite{RautenbachS13} %
gave a characterization of \openclosed-graphs using five forbidden induced subgraphs. 
Joos~\cite{Joos15} gave a characterization of mixed unit interval graphs without twins by an infinite class of forbidden induced subgraphs. 
Shuchat, Shull, Trenk, and West~\cite{Shuchat14} proved independently also this characterization, moreover, they complemented it by a quadratic-time algorithm that produces a mixed proper interval representation. 
Finally, Kratochv\'il and Talon~\cite{TalonK18} characterized the remaining classes.

Le and Rautenbach~\cite{Le12} characterized graphs that have a mixed unit interval representations in which all intervals have integer endpoints, and provided a quadratic-time algorithm that decides whether a given interval graph admits such a representation.

We refer the reader to the original papers for more details and concrete forbidden subgraphs.

Moreover, there are nice structural results of the subclasses of mixed unit interval graphs. For example, it is shown in \cite{RautenbachS13} that for \openclosed-intersection representations, open intervals are only really needed to represent claws, in particular, 
for any \openclosed-graph there exist a \openclosed-representation such that for every open interval, there is a closed interval with the same endpoints.
More structural results can be found in \cite{TalonK18}.

\begin{thm}[\cite{TalonK18}] \label{thm_mixed_repre_time}
The classes of semi-mixed and mixed unit interval graphs can be recognized in time ${\cal O}(n^2)$.
Moreover,  there exists an algorithm which takes a graph $G \in \semimixed$ on input, and outputs a corresponding \semimixed-representation of $G$ in time ${\cal O}(n^2)$. 
\end{thm}

\begin{cor}[\cite{TalonK18}] \label{claim_mixed_repre_time}
It is possible to modify the algorithm for semi-mixed unit interval graphs such that given a mixed unit interval graph $G$, it outputs a mixed unit interval representation of $G$ in time ${\cal O}(n^2)$.
\end{cor}
}

\section{Bubble model for mixed unit interval graphs}\label{s_Ububble}

\sv{
  \toappendix{\section{Additions to Section~\ref{s_Ububble}}}
}

In this section, we present a \U-bubble model, a new representation of mixed unit interval graphs which is inspired by the notion of bubble model for proper interval graphs created by Heggernes, Meister, and Papadopoulos~\cite{HeggernesMP09} in 2009.
\lv{\subsection{Definition of bubble model}
First, we present the bubble model for proper interval graphs as it was introduced by Heggernes~et~al.
}

\begin{definition}[Heggernes et al.~\cite{HeggernesMP09}, reformulated]
If $A$ is a finite non-empty set, then a \emph{2-di\-mensional bubble structure for $A$} is a partition $\B = \langle B_{i,j}\rangle_{1\le j\le k, 1\le i\le r_j}$, where 
$A=\bigcup_{i,j}%
{B_{i,j}}, $
$\emptyset \subseteq B_{i,j} \subseteq A$ for every $i, j$ with $1\le j\le k$ and $1\le i\le r_j$,  
 and $B_{1,1}\dots B_{k,r_k}$ are pairwise disjoint.
The \emph{graph given by \B}, denoted as $G(\B)$, is defined as follows:
\begin{enumerate}
\item $G(\B)$ has a vertex for every element in $A$, and
\item $uv$ is an edge of $G(\B)$ if and only if there are indices $i,i',j,j'$ such that $u \in \Bij,$ $v \in B_{i',j'},$ $|j-j'|\le 1$, and one of the two conditions holds: either $j=j'$ or $(i-i')(j-j')<0$.
\end{enumerate}

A \emph{bubble model} for a graph $G=(V,E)$ is a 2-dimensional bubble structure~\B for $V$ such that $G=G(\B).$ 
\end{definition}

\begin{thm}[Heggernes et al.~\cite{HeggernesMP09}]
A graph is a proper interval graph if and only if it has a bubble model.
\end{thm}

We define a similar matrix-type structure for mixed unit interval graphs where each set $B_{i,j}$ is split into four parts and edges are allowed also in the same row under specific conditions. 
\begin{definition}
Let $A$ be a finite non-empty set and $\B = \langle B_{i,j}\rangle_{1\le j\le k, 1\le i\le r_j}$ be a 2-dimensional bubble structure for $A$ such that 
$B_{i,j} = {\Bcc{i,j}} \cup\Bco{i,j}\cup\Boc{i,j}\cup\Boo{i,j}$, 
$B_{i,j}^{r,s}$ are pairwise disjoint, 
and $\emptyset \subseteq B_{i,j}^{r,s} \subseteq B_{i,j}$ for every $r,s\in\{+,-\}$ and $i, j$ with $1\le j\le k$ and $1\le i\le r_j$.
We call the partition \B a \emph{2-dimensional \U-bubble structure} for $A$.
\end{definition} 

We call  each set \Bij a \emph{bubble}, and each set $B_{i,j}^{r,s}, r,s\in\{+,-\},$ a \emph{quadrant} of the bubble \Bij. The \emph{type} of a quadrant $\Bij^{r,s}$, $r,s\in\{+,-\}$, is the pair $(r,s)$. 
We denote by $*$ both $+$ and $-$, for example $\Bhc{i,j}=\Boc{i,j}\cup\Bcc{i,j}$.
Bubbles with the same $i$-index form a \emph{row} of \B, and with the same $j$-index a \emph{column} of \B, we say vertices from bubbles $B_{i,1}\cup\ldots\cup B_{i,k}$ \emph{appear in row $i$}, and we denote $i$ as their \emph{row-index}.
We define an analogous notion for columns.
We denote the index of the first row with a non-empty bubble as $\btop{j}\df\min{\{i\mid \Bij \in \B \text{ and } B_{ij}\neq\emptyset\}}$. Thus, $\BT{j}$ is the first non-empty bubble in the column $j$.
Let $B$ be a bubble, then $\row{B}$ and $\col{B}$ is the row-index and column-index of $B$, respectively.
Let $u\in \Bij$, $v\in B_{i',j'}$; we say that $u$ is \emph{under} than $v$ and $v$ is \emph{above} $u$ if $i>i'$.

\tikzset{
  bubble/.style n args={4}{%
  }
}

\begin{figure}[h]
\centering

\begin{subfigure}[b]{\textwidth}
\centering

\begin{tikzpicture}[scale=1.3, yscale=1, bend angle=40, thick,place2/.style={thick,
inner sep=0pt,minimum size=1mm}, place/.style={circle,draw=black,fill=black,thick,
inner sep=0pt,minimum size=1mm}]

\def\dx{2}
\def\dy{0.69}

\node (i0) at (\dx+ 0,0) [place,label=below:{$i$}] {};
\node (k02) at (\dx+ 1-0.4,\dy) [place,label={$k_2$}] {};
\node (k01) at (\dx+ 1-0.8,\dy) [place,label={$k_1$}] {};
\node (k03) at (\dx+ 1,\dy) [place,label={$k_3$}] {};

\node (j0) at (\dx+ 1,0) [place,label=below:{$j$}] {};
\node (m0) at (\dx+ 1.5,\dy) [place,label=$m$] {};
\node (l01) at (\dx+ 2,0) [place,label={$\ell_1$}] {};
\node (l02) at (\dx+ 2,-\dy) [place,label=below:{$\ell_2$}] {};

\node (n0) at (\dx+ 2.5,\dy) [place,label=$n$] {};
\node (o0) at (\dx+ 3,0) [place,label=below:$o$] {};
\node (t0) at (\dx+ 3.5,-\dy) [place,label=below:$t$] {};
\node (p0) at (\dx+ 3,\dy) [place,label=$p$] {};
\node (s0) at (\dx+ 4,0) [place,label=below:$s$] {};
\node (u0) at (\dx+ 5,0) [place,label=below:$u$] {};

\node (q0) at (\dx+3.5,\dy) [place,color=red,label=\textcolor{red}{$q$}] {};
\node (r0) at (\dx+4,\dy) [place,color=red,label=above right:\textcolor{red}{$r$}] {};

\draw[fill,color={blue!35}] (\dx-1,-\dy/4) ellipse (0.17 and 1.5*\dy+0.1);

\node (c0) at (\dx-1,\dy) [place,color=blue,label=above right:\textcolor{blue}{{$c$}}] {};
\node (d0) at (\dx-1,\dy/2) [place,color=blue,label=above right:\textcolor{blue}{$d$}] {};
\node (e0) at (\dx-1,0) [place,color=blue,label=above right:\textcolor{blue}{$e$}] {};
\node (f0) at (\dx-1,-\dy/2) [place,color=blue,label=left:\textcolor{blue}{$f$}] {};
\node (g0) at (\dx-1,-\dy) [place,color=blue,label=left:\textcolor{blue}{$g$}] {};
\node (h0) at (\dx-1,-\dy*1.5) [place,color=blue,label=left:\textcolor{blue}{$h$}] {};

\node (a0) at (\dx-2,\dy) [place,color=blue,label=above:\textcolor{blue}{{$a$}}] {};
\node (b0) at (\dx-2,0) [place,color=blue,label=left:\textcolor{blue}{{$b$}}] {};

\foreach \x in {e0,f0,g0,h0}{
    \draw[color=blue] (i0)--(\x);
    }
\foreach \x in {b0,c0,d0}{
    \draw[color=blue] (a0)--(\x);
    }
\foreach \x in {c0,d0,e0,f0}{
    \draw[color=blue] (b0)--(\x);
    }


\draw (i0)--(j0);
\draw (k01)--(k02);
\draw (k01)--(k03);
\draw (k02)--(k03);
\draw (k01)--(j0);
\draw (k02)--(j0);
\draw (k03)--(j0);
\draw (l02)--(l01);
\draw (j0)--(l01);
\draw (j0)--(l02);
\draw (o0)--(l01);
\draw (o0)--(l02);
\draw (j0)--(m0);
\draw (n0)--(m0);
\draw (n0)--(l01);
\draw (n0)--(l02);
\draw (m0)--(l01);
\draw (m0)--(l02);
\draw (o0)--(p0);
\draw (o0)--(s0);
\draw (n0)--(o0);

\draw (t0)--(s0);
\draw (t0)--(o0);

\draw[color=red] (q0)--(r0);
\draw[color=red] (q0)--(r0);

\draw (u0)--(s0);
\foreach \x in {o0,s0,t0,p0} { 
    \draw[color=red] (q0)--(\x);
    \draw[color=red] (r0)--(\x);
     }

\draw[bend left=40]  (k01) to [bend left=40] (k03);

\end{tikzpicture}
\caption{Graph $G$; the blue ellipse denotes clique $cdefgh$; colors are used only for clarity.}
\end{subfigure}

\begin{subfigure}[b]{\textwidth}
\centering

\begin{tikzpicture}[scale=1, yscale=0.39, bend angle=40, thick,place2/.style={thick,
inner sep=0pt,minimum size=1mm}, place/.style={circle,draw=black,fill=black,thick,
inner sep=0pt,minimum size=1mm}]

\node  at (0,7) [align=left] {}; 

\draw[dotted, darkgray] (-1.2,6) -- (-1.2,-4);
\draw[dotted, darkgray] (0,6) -- (0,-4);
\draw[dotted, darkgray] (1,6) -- (1,-4);
\draw[dotted, darkgray] (2,6) -- (2,-4);
\draw[dotted, darkgray] (3,6) -- (3,-4);
\draw[dotted, darkgray] (4,6) -- (4,-4);
\draw[dotted, darkgray] (5,6) -- (5,-4);

\draw[{[-]}] (0,1) -- (1,1);
\draw[{(-)}] (1,4) --(2,4); 
\draw[{(-)}] (1,5) --(2,5); 
\draw[{(-)}] (1,3) -- (2,3);
\draw[{[-]}] (1,2) --(2,2); 
\draw[{[-)}] (2,1) -- (3,1);
\draw[{(-]}] (2,0) -- (3,0); 
\draw[{[-]}] (2,-1) -- (3,-1);
\draw[{[-]}] (2,-2) -- (3,-2);
\draw[(-)] (3,3)--(4,3); 
\draw[{[-]}] (3,2)--(4,2);
\draw[{[-]}] (4,1)--(5,1);
\draw[{[-)}] (4,0)--(5,0);
\draw[{[-]}] (5,2)--(6,2);

\draw[{[-]}] (3.3,-1)--(4.3,-1); 
\draw[{[-]}] (3.6,-2)--(4.6,-2); 

\draw[{[-]}] (-1.2,3) -- (-0.2,3);
\draw[{[-]}] (-1,2) -- (0,2);
\draw[{[-]}] (-0.8,0) -- (0.2,0);
\draw[{[-]}] (-0.6,-1) -- (0.4,-1);
\draw[{[-]}] (-0.4,-2) -- (0.6,-2);
\draw[{[-]}] (-2.2,2) -- (-1.2,2);
\draw[{[-]}] (-1.8,1) -- (-0.8,1);

\draw[{[-]}] (-1.4,-3) -- (-0.4,-3);

\node (i) at ( 0.5,1.4) [place2] {$i$};
\node (j) at ( 1.5,2.4) [place2] {$j$};
\node (k3) at ( 1.5,3.4) [place2] {$k_3$};
\node (k2) at ( 1.5,4.4) [place2] {$k_2$};
\node (k1) at ( 1.5,5.4) [place2] {$k_1$};
\node (m) at ( 2.5,1.4) [place2] {$m$};
\node (n) at ( 2.5,0.4) [place2] {$n$};
\node (l1) at ( 2.5,-0.6) [place2] {$\ell_1$};
\node (l2) at ( 2.5,-1.6) [place2] {$\ell_2$};
\node (p) at ( 3.5,3.4) [place2] {$p$};
\node (o) at ( 3.5,2.4) [place2] {$o$};
\node (s) at ( 4.5,1.4) [place2] {$s$};
\node (t) at ( 4.5,0.4) [place2] {$t$};
\node (u) at ( 5.5,2.4) [place2] {$u$};
\node (p) at ( 3.8,-0.6) [place2] {$q$};
\node (q) at ( 4.1,-1.6) [place2] {$r$};

\node (d) at ( -0.7,3.4) [place2] {$d$};
\node (e) at ( -0.5,2.4) [place2] {$e$};
\node (f) at ( -0.3,0.4) [place2] {$f$};
\node (g) at ( -0.1,-0.6) [place2] {$g$};
\node (h) at ( 0.1,-1.6) [place2] {$h$};

\node (a) at ( -1.7,2.4) [place2] {$a$};
\node (b) at ( -1.4,1.4) [place2] {$b$};
\node (c) at ( -0.9,-2.6) [place2] {$c$};

\node (foo) at ( 0,-3.5) [place2] {};

\end{tikzpicture}
\caption{A mixed unit interval representation of $G$.}
\end{subfigure}

\begin{subfigure}[b]{\textwidth}
\centering

\begin{tikzpicture}[scale=0.45, thick,place2/.style={thick,
inner sep=0pt,minimum size=1mm}, place/.style={circle,draw=black,fill=black,thick,
inner sep=0pt,minimum size=1.5mm}]

\node (a) at (0,12) [align=left] {}; 

\draw[line width=0.065cm] (-4,3) rectangle (-4+3,3+3);
\draw (-4,3+1.5) -- (-4+3,3+1.5);
\draw (-4+1.5,3) -- (-4+1.5,3+3);

\draw[{[-]}] (-3.75,5.25) -- (-2.75,5.25);
\draw[{[-)}] (-3.75,3.75) -- (-2.75,3.75);
\draw[{(-]}] (-2.25,5.25) -- (-1.25,5.25);
\draw[{(-)}] (-2.25,3.75) -- (-1.25,3.75);

\drawBubble{0}{8}{a}{}{}{}
\drawBubble{0}{4}{b}{}{}{}
\drawBubble{0}{0}{c}{}{}{}

\drawBubble{2}{8}{d}{}{}{}
\drawBubble{2}{6}{e}{}{}{}
\drawBubble{2}{4}{f}{}{}{}
\drawBubble{2}{2}{g}{}{}{}
\drawBubble{2}{0}{h}{}{}{}

\drawBubble{4}{6}{i}{}{}{}

\drawBubble{6}{6}{j}{}{}{k_i}

\drawBubble{8}{6}{\ell_i}{m}{n}{}

\drawBubble{10}{6}{o}{}{}{p}
\drawBubble{10}{4}{q}{}{}{}
\drawBubble{10}{2}{r}{}{}{}

\drawBubble{12}{6}{s}{t}{}{}

\drawBubble{14}{6}{u}{}{}{}

\end{tikzpicture}
\caption{A \U-bubble model of $G$ on the right, types of bubble quadrants on the left.}
\end{subfigure}
\caption{Three different representations of a mixed unit interval graph $G$.} \label{f_mixed_bubble_example}
\label{class-C}
\end{figure}
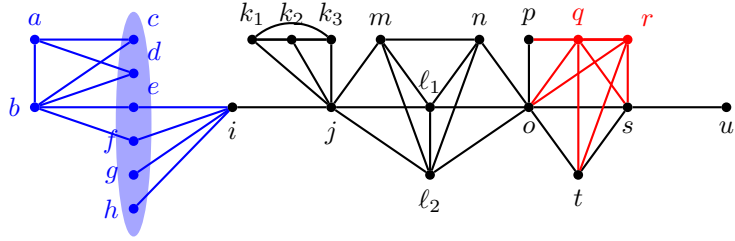
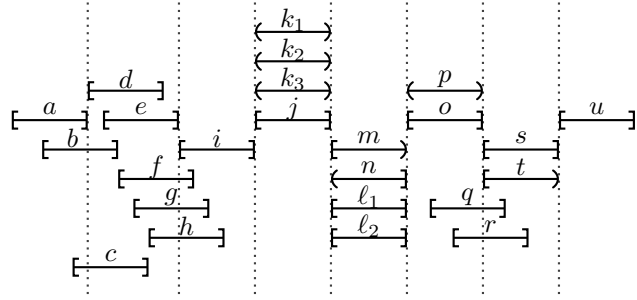
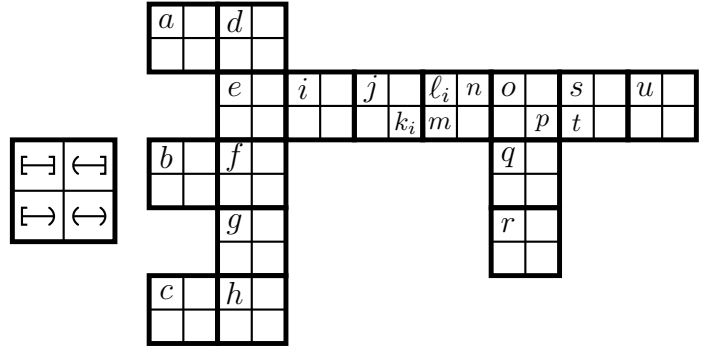

\FloatBarrier

\begin{definition}\label{d_bm_graph}
Let $\B = \langle B_{i,j}\rangle_{1\le j\le k, 1\le i\le r_j}$
be a 2-dimensional \mixed-bubble structure for $A$.
 The graph \emph{given by \B}, denoted as $G(\B)$, is defined as follows:
\begin{enumerate}
\item $V(G(\B)) = A$, %
\item $uv$ is an edge of $G(\B)$ if and only if there are indices $i,i',j,j'$ such that
  $u \in \Bij,$ $v \in B_{i',j'}$, or $v\in \Bij$, $u\in B_{i',j'}$, and one of the three conditions holds:
    \begin{enumerate}
    \item $j = j'$, or
    \item $j = j' - 1$ and $i > i'$, or %
    \item $j = j' - 1$ and $i = i'$ and $u\in \Bhc{i,j}, v\in\Bch{i',j'}$.%
    \end{enumerate}
\end{enumerate}
\end{definition}

The definition says that the edges are only between vertices from the same or consecutive columns and if $u\in\Bij$ and $v\in B_{i',j+1}$, there is an edge between $u$ and $v$ if and only if $u$ is lower than $v$ ($i > i'$), or they are in the same row and $u\in\Bhc{i,j},v\in\Bch{i',j+1}$.
\sv{
Vertices from the same column form a clique, as well as vertices from the same bubble. Vertices from the same bubble are almost-twins and their neighborhoods can differ only in the same row, anywhere else they behave like twins. Vertices from the same bubble quadrant are true twins.
}

\lv{
\begin{observation} \label{r-complete}
Vertices from the same column form a clique, as well as vertices from the same bubble. Moreover, vertices from the same bubble are almost-twins and their neighborhoods can differ only in the same row, anywhere else they behave like twins. Vertices from the same bubble quadrant are true twins.
\end{observation}
}

\begin{definition}\label{d_u_bubblemodel}
Let $G=(V,E)$ be a graph. A \emph{\U-bubble model for a graph} $G$ is 
a {2-dimensional \mixed-bubble structure $\B = \langle B_{i,j}\rangle_{1\le j\le k, 1\le i\le r_j}$ for $V$} such that
\begin{enumerate}[(i)] 
\item $G$ is isomorphic to $G(\B)$, and 
\item each column and each row contains a non-empty bubble, and
\item no column ends with an empty bubble, and%
\item $\btop{1}=1$, and for every $j\in\{1,\dots,k-1\}: \btop{j}\le\btop{j+1}.$
\end{enumerate} 
\end{definition}

For a \U-bubble model $\B = \langle B_{i,j}\rangle_{1\le j\le k, 1\le i\le r_j}$, 
by the \emph{number of rows} of $\B$ we mean $\max\{r_j\mid 1\le j\le k\}$. We define 
the \emph{size} of the \U-bubble model \B as the number of columns multiplied by the number of rows, i.e., 
$k\cdot \max\{r_j\mid 1\le j\le k\}.$

See Figure~\ref{f_mixed_bubble_example} with an example of a mixed unit interval graph, given by a mixed unit interval representation, and by a \U-bubble model.

\newcommand{\curr}{\ensuremath{\mathsf{curr}}\xspace}
\newcommand{\Top}[1]{\ensuremath{C^\mathsf{{top}}_{#1}}\xspace}
\newcommand{\level}{\ensuremath{\mathsf{L}}\xspace}
\newcommand{\nextP}{\ensuremath{{next^{P}}}}
\newcommand{\bubbleset}{\ensuremath{\mathscr{B}}}
\lv{\subsection{Construction of \U-bubble model}}
First, we construct a mixed unit interval representation \I of a graph $G$ using the quadratic-time algorithm \sv{by~\cite{TalonK18}}\lv{see Corollary~\ref{claim_mixed_repre_time}}; 
then each vertex of $G$ is represented by a corresponding interval in \I. Having a mixed unit interval representation of the graph, our algorithm outputs a \U-bubble model for the graph in $\Oh{n}$ time.

  Given a mixed unit interval representation \I, we put all intervals (vertices) that are almost-twins in \I into a single bubble, to the particular quadrant which corresponds by its type to the type of the interval.
  From now on, we speak about bubbles only, we denote the set of all such bubbles by \bubbleset. We are going to determine their place (row and column) to create a 2-dimensional \U-bubble structure for \bubbleset. We show that the \U-bubble structure is a \U-bubble model for our graph.
Based on the order~$\sigma$ by endpoints of intervals in the representation \I from left to right, we obtain the same order on bubbles in \bubbleset.
The idea of the algorithm is to process the bubbles in the order~$\sigma$, and assign to each bubble
  its column immediately after processing it. 
  During the processing, the algorithm maintains an auxiliary path in order to assign rows at the end.
  Thus, rows are assigned to each bubble after all bubbles are processed.

For bubbles $A, B\in \bubbleset$, $A<_\sigma B$ denotes that $A$ is smaller than $B$ in order $\sigma.$ 
We denote the order of bubbles by subscripts, i.e., $B_1<_\sigma B_2<_\sigma \ldots$ are all bubbles in the described order $\sigma$. 
  For technical reasons, we create two new bubbles: $B_{start}$, $B_{end}$ such that $\li{B_{start}}=\ri{B_{start}}=-\infty$. We refer to them as \emph{auxiliary bubbles}, in particular, if we speak about bubbles, we exclude auxiliary bubbles. 
 We enhance the representation in a way that each bubble $B\in\bubbleset$ has a pointer 
 $prev:\bubbleset\to \bubbleset\cup\{B_{start}\}$ defined as follows.
  \sv{\vspace{-5pt}}
  \begin{align*}
  prev(B)=
      \begin{cases}
        B_{start} & \text{if } \li{B}<\ri{B_1},\\
        A \text{~~~~\sv{s.\,t. }\lv{such that} } \li{B}=\ri{A} & \text{if such a bubble $A$ exists},\\
        B_j \text{~~~\sv{s.\,t. }\lv{such that} } j=\max_i \bigl\{i \mid \li{B}>\ri{B_i}\bigr\} & \text{otherwise}.
    \end{cases}
   \end{align*}
  \sv{\vspace{-12pt}}

In order to set rows at the end, the algorithm is creating a single oriented path $P$ that has the necessary information about the height of elements in the \U-bubble structure being constructed. 
Some of the arcs of the path can be marked with level indicator (\level).
For ease of notation, we use $\nextP(B_i)=B_j$ to say that $B_j$ is the next element on path $P$ after $B_i$. Note that we can view $P$ as an order of bubbles; we denote by $A<_P B$, $A,B\in \bubbleset,$  the information that $A$ occurs earlier than $B$ on $P$.
Also from technical reasons,  $P$ starts and ends with $B_{start}$ and $B_{end}$, respectively.
Except $P$ and pointers $prev$ and $\nextP$, the algorithm remembers the highest bubble of column $i$, denoted by $\Top{i}$. Also, denote by \curr, the index of the currently processed column. 

Now, we are able to state the algorithm for assigning columns and rows to bubbles in \bubbleset ~and its properties  which will be useful for showing the correctness.

\vspace{3pt}
\begin{compactenum}[\bf Property 1:]
\item  Bubbles are processed (and therefore added somewhere to $P$) one by one respecting the order $\sigma$. 
    \label{bm_constr_prop_oneByOne}
\item  The order induced by $P$ of already processed vertices never changes, i.e., once $A\le_P B$ then $A\le_P B$ for the rest of the algorithm.
    \label{bm_constr_prop_relativePOrder}
 \item The arc of $P$ between bubbles $A$ and $B$ has the level indicator (\level) if and only if $\ri{A}=\li{B}$. Moreover, if the arc from $A$ to $B$ has level indicator, then $\col{A}<\col{B}$. 
    \label{bm_constr_prop_levelInd}
  \item $\col{B_i}\le\col{B_j}$ whenever $i\le j$.
    \label{bm_constr_prop_columns}
  \item $prev(B)$ is the closest ancestor of $B$ on $P$ 
  in the previous column, i.e., $prev(B)=\max\{A \mid A\le_P B, \col{A}=\col{B}-1\}$. 
    \label{bm_constr_prop_prev}
  \item The order induced by $P$ of vertices in the same column is exactly the order of those vertices induced by $\sigma$.
    \label{bm_constr_prop_PSigmaColumn}
\end{compactenum}

\lv{\subsection{Algorithm}}
\sv{\vspace{3pt}\noindent{\bf Algorithm.~~}}
Given bubbles  $B_1, B_2,\dots$ in \bubbleset~  ordered by $\sigma$, the algorithm creates $P$ by processing bubbles one by one in order $\sigma$. The algorithm outputs a row and a column to each bubble. 
Initially, set $\col{B_1}=1$, $P=\{B_{start},B_1,B_{end}\}$, \curr=1 and $\Top{1}=B_1$.
\lv{%

}%
Suppose that $i-1$ bubbles have been already processed, for $i\geq 2$. 
Split the cases of processing bubble $B_{i}$ based on the following possibilities:
\begin{enumerate}[i.]
  \item ${\li{B_i}>\ri{\Top{\curr}}}$: First increase \curr by one, then set $\col{B_i}=\curr$ and $\Top{\curr}=B_i$.

\item ${\li{B_i}=\ri{\Top{\curr}}}$:
  First increase \curr by one, then set $\col{B_i}=\curr$ and $\Top{\curr}=B_i$.
  Let $Q$ be $\nextP(\Top{\curr-1})$.
  Substitute arc in $P$ from $\Top{\curr-1}$ to $Q$ with two new arcs $\Top{\curr-1}$ to $B_i$ that has \level indicator set and from $B_i$ to $Q$.

\item ${\li{B_i}<\ri{\Top{\curr}}}$:
  Set $\col{B_i}=\curr$.
  \end{enumerate}

  We continue only with cases i. and iii. %
  and distinguish multiple possibilities:
  \begin{enumerate}[1.]
    \item ${\ri{prev(B_i)}=\li{B_i}}$:
      Let $Q$ be $\nextP(prev(B_i))$.
      Then substitute arc  in $P$ from $prev(B_i)$ to $Q$ with two new arcs $prev(B_i)$ to $B_i$ that has \level indicator set and from $B_i$ to $Q$. %
    \item ${\ri{prev(B_i)}<\li{B_i}}$:
      And split this case further based on the properties of $B_{i-1}$.
    \begin{enumerate}[2a.]
        \item ${prev(B_{i-1})=prev(B_{i})}$:
          Let $Q$ be $\nextP(B_{i-1})$.
          Substitute arc in $P$  from $B_{i-1}$ to $Q$ with two new arcs $B_{i-1}$ to $B_i$ and from $B_i$ to $Q$. %
        \item ${prev(B_{i-1})\neq prev(B_{i})}$:
          Let $Q$ be $\nextP(prev(B_i))$.
          Then substitute arc  in $P$ from $prev(B_i)$ to $Q$ with two new arcs $prev(B_i)$ to $B_i$ and from $B_i$ to $Q$.
      \end{enumerate}
  \end{enumerate}

Now, assign rows to bubbles by a single run over $P$, inductively:  
Take the first bubble $B$ of $P$ and assign $\row{B}\df 1$. 
Let $B$ be the last bubble on $P$ with already set row index. We are about to determine $\row{\nextP(B)}$.
If arc  in $P$ from $B$ to $\nextP(B)$ has \level indicator, set  $\row{\nextP(B)}\df\row{B}$, otherwise $\row{\nextP(B)}\df\row{B}+1$.

\toappendix{
  \subsection{Correctness}\label{sub:correctness}
Here, we show that the algorithm above gives us a \U-bubble model for a graph given by mixed unit interval representation. It gives us the forward implication of Theorem~\ref{thm_bubble}.

\begin{lemma}\label{lem_lin_construction}
Given a mixed unit interval representation \I of a connected graph $G$ on $n$ vertices, the \U-bubble model can be constructed in $\Oh{n}$ time.
\end{lemma}
\begin{proof}[Proof of Lemma~\ref{lem_lin_construction}]
We show the correctness of the construction, i.e., that the constructed object satisfies Definition~\ref{d_u_bubblemodel} and that declared Properties 1-6 are satisfied during the whole algorithm.
 It follows immediately from the construction that Properties \ref{bm_constr_prop_oneByOne}--\ref{bm_constr_prop_columns} are satisfied.
Observe that $prev(B)$ is always in the previous column than $B$, for $B\in\bubbleset$. Moreover, observe that in step ii.~of the algorithm, $\Top{\col{B_i}-1}=prev(B_i).$ Then, Property~\ref{bm_constr_prop_prev} follows from the construction. 
Property~\ref{bm_constr_prop_PSigmaColumn} can be seen by examining the construction. Let $A,B$ be two bubbles in the same column such that $A <_\sigma B$.
Either $prev(A) = prev(B)$, then $B$ is put later than $A$ on $P$. Or $prev(A)<_\sigma prev(B)$, then, by the construction, $prev(B)$ is put after $A$ and $B$ is put after $prev(B)$. In both cases, $A<_P B$. Using Property~\ref{bm_constr_prop_relativePOrder}, the Property~\ref{bm_constr_prop_PSigmaColumn} holds.

Let \B be the \U-bubble structure for \bubbleset~output by the construction above and $G(\B)$ be a graph given by \B. We show that $\B$ is a \U-bubble model for $G$.
Parts (ii), (iii) from Definition~\ref{d_bm_graph} are clearly satisfied.
It remains to show (i) and (iv).

Let us start with (i), that is $G(\B)$ is isomorphic to $G$.

Let $u\in B_i$, $v\in B_j$. Recall that $\li{u}=\li{v}$ if and only if $B_i=B_j$. Since this case is trivially satisfied, without loss of generality, we assume $B_i<_\sigma B_j$.
We distinguish a few cases based on the position of $B_i$ and $B_j$ in \B.

First, let $B_i$ and $B_j$ be in nonconsecutive columns in \B. Denote by $c=\col{B_i}$. By the definition, $u$ and $v$ are nonadjacent in $G(\B)$. 
By the construction, there exists a nonempty bubble $\Top{c+1}$ in \B such that it is the top bubble of column $c+1$. 
It follows that $\Top{c+1}>_\sigma B_i$, by Property~\ref{bm_constr_prop_columns}, and also $\Top{c+1}\neq B_i$.
Since the construction assigns $B_j$ to a different column than $\Top{c+1}$, we know that $\ri{\Top{c+1}}\le\li{B_j}$. It gives immediate conclusion that $u,v$ are not adjacent in $G$.

Second, let $B_i$ and $B_j$ be in the same column $c$ in \B. 
Vertices $u,v$ are adjacent by the definition in $G(\B)$. 
By the construction, there exists a nonempty bubble $\Top{c}$ in \B such that it is the top bubble of the same column and 
$\li{\Top{c}}<\li{B_i}<\li{B_j}<\ri{\Top{c}}=1+\li{\Top{c}}$. 
Therefore, $u,v$ are adjacent in $G$.

Third, let $B_i$ and $B_j$ appear in consecutive columns in \B. 
We denote $c=\col{B_i}=\col{B_j}-1$.
By the definition, vertices $u,v$ are adjacent in $G(\B)$ if and only if either $\row{B_i}>\row{B_j}$, or $\row{B_i}=\row{B_j}$ and $u\in B_i^{*+}, v\in B_j^{+*}$.
By Properties \ref{bm_constr_prop_oneByOne} and \ref{bm_constr_prop_relativePOrder} of $P$, it is sufficient to verify only the situation when bubble $B_j$ was added.
Observe that if $B<_P B'$ then $\row{B}\le\row{B'}$. We split the case into the following all possibilities:
\begin{enumerate}[$\bullet$]
\item $prev(B_j)>_\sigma B_i$: By the definition of $prev$ and the interval property, $u$ is non-adjacent to $v$.
 By Property~\ref{bm_constr_prop_prev}, $prev(B_j)<_P B_j$. By Property~\ref{bm_constr_prop_PSigmaColumn}, $B_i<_P prev(B_j)$. Since $\ri{B_i}\neq\li{B_j}$, by Property~\ref{bm_constr_prop_levelInd}, $\row{B_i}<\row{B_j}$.
\item $prev(B_j)= B_i$: By Properties~\ref{bm_constr_prop_levelInd}, \ref{bm_constr_prop_prev} and the rows assignment, $\row{B_j}=\row{B_i}$ if and only if $\li{B_j}=\ri{B_i}$.
        Therefore, there is an edge in both models if and only if  $u$ and $v$ are of correct type; that is $u$ has type $(*,+)$ and $v$ has type $(+,*)$.
\item $prev(B_j)<_\sigma B_i$: By the definition of $prev$ and the interval property, $u$ is adjacent to $v$ and $\ri{B_i}>\li{B_j}$. By Property~\ref{bm_constr_prop_PSigmaColumn}, $prev(B_j)<_P B_i$. By Property~\ref{bm_constr_prop_prev} and the rows assignment, $\row{B_i}\ge\row{B_j}$. By Property~\ref{bm_constr_prop_levelInd}, the equality cannot occur. Therefore, $\row{B_i}>\row{B_j}$.
\end{enumerate}

Part (iv) follows by the construction of $P$.
When $B=\Top{j},j\ge 2$ is added on $P$, by Property~\ref{bm_constr_prop_prev}, $prev(B)<_P B$. Note that $\Top{j-1} \le_P prev(B)$. We obtain $\row{\Top{j-1}}\le\row{\Top{j}}$ for every possible $j$. 
Also note that $\row{B_1}=\row{\Top{1}}=1$.

It remains to show the running-time of the algorithm. 
Note that $prev$ can be easily computed by a single run over the representation, as well as the assigning columns can be done simultaneously by a single run over the representation (having $prev$ and remembering top bubbles of columns). Moreover, rows of the vertices are assigned by a single run over path $P$ which leads to overall running time $\Oh{n}$ where $n$ is the number of intervals of the given mixed unit interval representation.
\end{proof}

}

\lv{
\subsection{Proof of Theorem \ref{thm_bubble}}

}
\sv{
  \medskip
}%
\begin{proof}[Proof of Theorem~\ref{thm_bubble}]
\sv{The forward implication follows from the above algorithm. Its correctness is proved in the full version of the paper.
Second}%
\lv{First}, we prove the reverse implication: given a \U-bubble model for a graph $G$, we construct a mixed unit interval representation of $G$.
Let $\B = \langle B_{i,j}\rangle_{1\le j\le k, 1\le i\le r_j}$ be a \U-bubble model of $G$. Let
\lv{$$\varepsilon\df\frac{1}{\max{\{r_j\mid 1\le j \le k\}}}.$$}
\sv{$\varepsilon\df{{1}/({\max{\{r_j\mid 1\le j \le k\}}}}).$}
We create a mixed unit interval representation $\I$ of $G$ as follows.
Let $v\in B_{i,j}^{r,s}$, where $r,s\in\{+,-\}$. The corresponding interval $I_v$ of $v$ has the properties: 
\lv{$$I_v\in\I^{r,s}\text{ and }\li{I_v}\df j+(i-1)\varepsilon.$$}
\sv{$I_v\in\I^{r,s}\text{ and }\li{I_v}\df j+(i-1)\varepsilon.$}
Note that all vertices from the same bubble are represented by intervals that are almost twins \lv{(they have the same left ends) }and the type of an interval corresponds with the type of the bubble quadrant.
Since $\varepsilon$ was chosen such that $\varepsilon(i-1) < 1$ for any row $i$ in $\B$,
the graph given by the constructed mixed unit interval representation is isomorphic to the graph given by \B.
\lv{%

}%
\lv{The forward implication follows from Lemma~\ref{lem_lin_construction}.}%
\end{proof}

\toappendix{
\subsection{Properties of \U-bubble model}\label{s_bubble_prop}

In this section, we give basic properties of a \U-bubble model which are used later in the text. It is readily seen that a \U-bubble model of graph $G=(V,E)$ has at most $n$ rows and $n$ columns where $n$ is the number of vertices of $G$ since each column and each row contains at least one vertex. Consequently, the size of a \U-bubble model is at most $n^2.$

Two basic characterizations of a graph are the size of a maximum clique and the size of a maximum independent set in the graph. The problem of finding those numbers is \NP-complete in general but it is polynomial-time solvable in interval graphs.
 We show a relation between those two numbers and the size of a \U-bubble model for the graph. We start with the size of a maximum independent set.

\begin{lemma}\label{l_bubble_mis_columns}
Let $G$ be a mixed unit interval graph, and let \B be a \mixed-bubble model for $G$. The number of columns of \B is at least \Mis{G} and at most $2\Mis{G}$.
\end{lemma}
\begin{proof}
Let $I$ be a maximum independent set of $G$, and let $k$ be the number of columns of \B. We have that $\Mis{G} \ge \lceil k/2 \rceil$ from the property that two non-consecutive columns from \B are not adjacent in \GB.
Since each column forms a clique, only one vertex from each column can be in $I$. Therefore, $\Mis{G} \le k.$
\end{proof}

In the bubble model for unit interval graphs, $\Mis{G}$ is equal to the number of columns \cite{HeggernesMP09}.
However, the gap in Lemma~\ref{l_bubble_mis_columns} cannot be narrowed in general---consider an even number $k$ and the following unit interval graphs: path on $k$-vertices ($P_k$) and a clique on $k$ vertices ($K_k$). There exists a unit interval representation of $P_k$ using only closed intervals which leads to a \mixed-bubble model of $P_k$ containing one row and $k$ columns, where $\Mis{P_k}=\lceil k/2 \rceil$. A \mixed-bubble model of $K_k$ contains $k$ rows and one column, where $\Mis{K_k}=1=\text{number of columns}$. 
 
Another important and useful property of graphs is the size of a maximum clique. We show that a maximum clique of a mixed unit interval graph can be found in two consecutive columns of a \mixed-bubble model of the graph, see Figure~\ref{f_max_clique}. 

\if\povolenyobrazky1
\begin{figure}[bt]
        \begin{center}
        \includegraphics{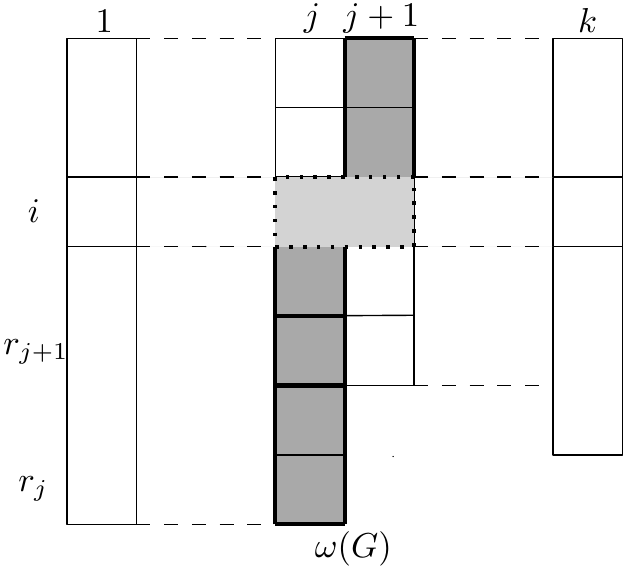}
        \end{center}
        \caption{A maximum clique of $G$ in a \U-bubble model. Dark grey color~represents the bubbles that are fully contained in the clique. Light grey color highlights two bubbles where only parts of them are contained in the clique, concretely the one of the sets $B_{i,j}, B_{i,j+1}$, and $B_{i,j}^{*+}\cup B_{i,j+1}^{+*}$ with the maximum size.}\label{f_max_clique}
\end{figure}
\fi

\begin{lemma}\label{l_bubble_clique_columns}
Let $G$ be a mixed unit interval graph, and let \B be a \mixed-bubble model for $G$. Then
the size of a  maximum clique is

\begin{align*}%
\Mclique{G} = \max_{\substack{j\in\{1, \dotsc, k-1\} \\i\in\{1,\dotsc,r_{j+1}\}}}
&{\left( 
\sum_{i'=i+1}^{r_j}|B_{i',j}| + \sum_{i'=1}^{i-1}|B_{i',j+1}| + a_i\right)
},\\
 &a_i = 
 \begin{cases}
 \max{\left\{
 |B_{i,j}|,
 |B_{i,j+1}|, 
 |B_{i,j}^{*+}| + |B_{i,j+1}^{+*}|
 \right\}} & i\le r_j,\\
 |B_{i,j+1}| & \text{otherwise.} 
\end{cases}
\end{align*}
\end{lemma}
\begin{proof}
  Let $K$ be a maximum clique of $G$. Notice, $K$ does not contain two vertices from nonconsecutive columns, as there are no edges between nonconsecutive columns. Furthermore,  vertices $u$ and $v$ from two consecutive columns $C_j$ and $C_{j+1}$, respectively,  can be in $K$ only if $u$ is under $v$ or they are in the same row in quadrants of types $\{*+\}$ and $\{+*\}$, respectively.  
   
  On the other hand, vertices from one column of $\B$ create a clique in $\GB$. Moreover, if we split any two consecutive columns $C_j$ and $C_{j+1}$ in row $i$ (for any index $i\in\{1,\ldots,\min{\{r_j,r_{j+1}\}}\}$), the second part of $C_j$ with the first part of $C_{j+1}$ form a clique. This is true even together with bubble quadrants $B_{i,j}^{*+} \cup B_{i,j+1}^{+*}$.   
\end{proof}
}

\section{Maximum cardinality cut}\label{s_maxcut}

\sv{
  \toappendix{\section{Additions to Section~\ref{s_maxcut}}\label{a_maxcut}}
}

This section is devoted to the time complexity of the \maxcut problem on (mixed) unit interval graphs. 
\lv{%

}%
\lv{\subsection{Notation}}
A \emph{cut} of a graph $G(V,E)$ is a partition of $V(G)$ into two subsets $S, \overline S$, where $\ol S=V(G)\setminus S$. Since $\ol S$ is the complement of $S$, we say for the brevity that a set $S$ is a cut and similarly we use terms \emph{cut vertex} and \emph{non-cut} vertex for a vertex $v\in S$ and $v\in\ol S$, respectively.
The \emph{cut-set} of cut $S$ is the set of edges of $G$ with exactly one endpoint in $S$, we denote it \cutset. Then, the value $|\cutset|$ is the \emph{cut size} of $S$. A \emph{maximum (cardinality) cut} on $G$ is a cut with the maximum size among all cuts on $G$. We denote the size of a maximum cut of $G$ by $\mcs{G}$. Finally, the \maxcut problem is the problem of determining the size of the maximum cut.

\subsection{Time complexity is still unknown on unit interval graphs}\label{s_example}

As it was mentioned in the introduction, there is a paper \emph{A polynomial-time
algorithm for the maximum cardinality cut problem in proper interval graphs} by Boyaci, Ekim, and Shalom from 2017 \cite{BoyaciES17}, claiming that the \maxcut problem is polynomial-time solvable in unit interval graphs and giving a dynamic programming algorithm based on the bubble model representation.
{We realized that the algorithm is incorrect; this section is devoted to it.}

We start with a counterexample to the original algorithm.
\begin{example}
Let $\B = \langle B_{i,j}\rangle_{1\le j\le 2, 1\le i\le 2}$, where $B_{1,1} = \{v_1\}$, $B_{2,1}=\{v_2\}$, $B_{1,2}=\{v_3,v_4,v_5\}$, $B_{2,2}=\{v_6\}$, be a bubble model for a graph $G$, see also Figure~\ref{p_counterexample}.
In other words, this bubble model corresponds to a unit interval graph on vertices $v_1, v_2,v_3,v_4,v_5,v_6$ where there is an edge $v_1v_2$, and vertices $v_2,v_3,v_4,v_5,v_6$ create a complete graph without an edge $v_2v_6$.

\if\povolenyobrazky1
\begin{figure}[bt]
        \begin{center}
        \includegraphics{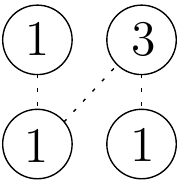}
        \end{center}
        \caption{A counterexample to the original algorithm, a bubble model \B where the numbers denote the number of vertices in each bubble, and dashed lines indicate the edges between bubbles.}\label{p_counterexample}
\end{figure}
\fi

Then, according to the paper \cite{BoyaciES17}, the size of a maximum cut in $G$ is eight.
To be more concrete, the algorithm from \cite{BoyaciES17} fills the following values of dynamic table: $F_{0,1}(0,0) = 4$, $F_{2,1}(1,1) = 8$ for $s_{2,1} = 1, s_{2,2}=1$, and finally, $F_{0,0}(0,0) = 8$ which is the output of the algorithm.
However, the size of a maximum cut in $G$ is only seven. 
Suppose, for contradiction, that the size of a maximum cut is eight. 
As there are ten edges in total in $G$, at least one vertex of the triangle $v_3,v_4,v_5$ must be a cut-vertex and one not. Then, those two vertices have three common neighbors. Therefore, the size of a maximum cut is at most seven which is possible; for example, $v_1,v_4,v_5$ are cut-vertices.
\end{example}

The brief idea of the algorithm in \cite{BoyaciES17} is to process the columns from the biggest to the lowest column from the top bubble to the bottom one. 
Once we know the number of cut-vertices in the actual processed bubble $B$ (in the column $j$)  and the number of cut-vertices which are above $B$ in the columns $j$ and $j+1$, we can count the exact number of edges. 
For each bubble and each such number of cut-vertices in the columns $j$ and $j+1$ (above the bubble), we remember only the best values of \maxcut\footnote{We refer the reader to the paper \cite{BoyaciES17} for the notation and the description of the algorithm.}.

We claim that the algorithm and its full idea from \cite{BoyaciES17} are incorrect since we lose the consistency there---to obtain a maximum cut, we do not remember anything about the distribution of cut vertices within bubbles, that was used in the previously processed column. Therefore, there is no guarantee that the final outputted cut of the computed size exists. 
To be more specific, one of two problems is in the moving from the column $j$ to the column $j-1$ since we forget there too much.  
The second problem is that for each bubble $B_{i,j}$ and for each possible numbers $x, x'$ we count the size $F_{i,j}(x,x')$ of a specific cut 
 and we choose some values $s_{i,j}$, $s_{i,j+1}$ (possibly different; they represents the number of cut-vertices in the bubbles $\Bij, B_{i,j+1}$) which maximize the values of $F_{i,j}(x,x')$. In few steps later, when we are processing the bubble  $B_{i,j-1}$, again, for each possible values $y$ and $y'$ we choose some values $s'_{i,j-1}$ and $s'_{i,j}$ such that they maximize the size of $F_{i,j-1}(y, y')$. However, we need to be consistent with the selection in the previous column, i.e., to guarantee that $s_{i',j} = s_{i,j}$ for any particular values $y$, $y'=x,$ and $x'$.

A straightforward correction of the algorithm would lead to remembering too much for a polynomial-time algorithm. 
However, we can be inspired by it to obtain a subexponential-time algorithm. 
We attempted to correct the algorithm or extend the idea leading to the polynomiality.
However, despite lots of effort, we were not successful and it seemed to us that the presented algorithm is hardly repairable.
\lv{%
 We note here, that there is another paper by the same authors~\cite{BoyaciES18} where a very similar polynomial algorithm is used for \maxcut of co-bipartite chain graphs with twins. Those graphs can be viewed as graphs given by bubble models with two columns; but having two columns is a crucial property for the algorithm.

}
To conclude, the time complexity of the \maxcut problem on unit interval graphs is still not resolved and it seems to be a challenging open question. 

\subsection{Subexponential algorithm  in mixed unit interval graphs}\label{s_subexp}
Here, we present a subexponential-time algorithm for the \maxcut problem in mixed unit interval graphs.
Our aim is to have an algorithm running in $2^{\tilde{{\cal O}}(\sqrt{n})}$ time.
Some of the ideas, for unit interval graphs, originated in discussion with Karczmarz, Nadara, Rzazewski, and Zych-Pawlewicz at Parameterized Algorithms Retreat of University of Warsaw 2019~\cite{karpacz}.

Let us start with a notation. 
Let $G$ be a graph, $H$ be a subgraph of $G$, and $S$ be a cut of $H$, we say that a cut $X$ of $G$ \emph{agrees with} $S$ in $H$ if $X=S$ on $H$. 
Let $G$ be a mixed unit interval graph. 
We take a \U-bubble model $\B = \langle B_{i,j}\rangle_{1\le j\le k, 1\le i\le r_j}$ for $G$
 and we distinguish columns of \B according to their number of vertices.
We denote by $b_{ij}$ the number of vertices in bubble $\Bij$ and by $c_j$ the number of vertices in column $j$, i.e., $b_{ij}=|\Bij|$ and $c_j = \sum_{i=1}^{r_j}{b_{i,j}}$. We call a column $j$ with $c_j>\sqrt{n}$ a \emph{heavy column}, otherwise a \emph{light column}. We call consecutive heavy columns and their two bordering light columns a \emph{heavy part} of \B (if \B starts or ends with a heavy column, for brevity, we add  an empty column at the beginning or the end of \B, respectively), and we call their light columns \emph{borders}. Heavy part might contain no heavy columns in the case that two light columns are consecutive.

 Note that we can guess all possible cuts in one light column without exceeding the aimed time and that most of those light column guesses are independent of each other---once we know the cut in the previous column, it does not matter what the cut is in columns before. 
 Furthermore, there are at most $\sqrt{n}$ consecutive heavy columns which allow us to process them together. More formally, we show that we can determine a maximum cut independently for each heavy part, given a fixed cut on its borders, 
 as stated in the following lemma.
 \sv{The formal proof is in the full version.} %

\begin{lemma} \label{l_maxcut_divideAndConquer}
Let $G$ be a mixed unit interval graph and $\B$ be a \U-bubble model for $G$ partitioned into heavy parts $\hat\B_1, \cdots, \hat\B_{p}$ in this order.
If $S = S_0\cup\cdots\cup S_{p}$ is a (fixed) cut of light columns $C_0,\dots,C_p$ in G(\B) such that $S_j$ is a cut of $C_j$, $j\in\{0,\dots,p\}$, then the size of a maximum cut of $G$ that agrees with $S$ in light columns is
\vspace{-10pt}
$$
\mcs{G,S} = \sum_{j=1}^{p}\mcs{G(\hat\B_j),S_{j-1}\cup S_{j}} - (\sum_{j=1}^{p-1}{|S_j|\cdot |C_j\setminus S_j|}) 
\vspace{-7pt}
$$
where $\mcs{G(\hat\B_j),S_{j-1}\cup S_{j}}$ denotes the size of a maximum cut of $G(\hat\B_j)$ that agrees with $S_{j-1}\cup S_{j}$ in its borders $C_{j-1},C_{j}$.
\end{lemma}
\toappendix{
  \lv{\begin{proof}}
    \sv{\begin{proof}[Proof of Lemma~\ref{l_maxcut_divideAndConquer}]}
It is readily seen that once we have a fixed cut in an entire column $C$ of a bubble model, a maximum cut of columns which are to the left of $C$ (including $C$) is independent on a maximum cut of those which are to the right of $C$ (including $C$). Therefore, we can sum the sizes of maximum cuts in heavy parts which are separated by fixed cuts. 
However, the cut size of middle light columns is counted twice since they are contained in two heavy parts. Therefore, we subtract them.
\end{proof}
}

Now, our aim is to determine the size of a maximum cut for a heavy part $\hat\B$ given a fixed cut on its borders. 
Note that if $\hat \B$ is a heavy part with no heavy columns, we can straightforwardly count the number of cut edges of $G(\hat\B)$, i.e., $\mcs{G(\hat\B)}$, assuming a fixed cut on borders is given. Therefore, we are focusing on a situation where at least one heavy column is present in a heavy part. 
We use dynamic programming to determine the size of a maximum cut on each such heavy part.

First, we present a brief idea of the dynamic programming approach, followed by technical definitions and proofs later. 
We take bubbles in $\hat\B$ which are not in borders and process them\break one-by-one in top-bottom, left-right order. 
When processing a bubble, we consider all the possibilities of numbers of cut-vertices in each its quadrant.
We refer to the already processed part after $i$-th step as $G_i$, that is $G_i$ is the induced subgraph of $G(\hat\B)$ with 
$V(G_i)=B_1\cup\cdots\cup B_i\cup C_0 \cup C_{l+1}$ 
where 
$C_0$ and $C_{l+1}$ are borders of $\hat\B$ and $B_j$, $j\in\{1,\dots,i\}$ are first $i$ bubbles in top-bottom, left-right order in $\hat\B$ (as it is shown in Figure~\ref{f_Gi}).

We store all possible $(l+1)$-tuples $(s_1,s_2,\dots,s_l,a)$, where $l$ is the number of heavy columns, $s_j$ characterizes the number of all cut vertices in the $j$-th heavy column, and number $a$ characterizes the number of cut vertices of types $(*,+)$ in the last processed bubble.
\lv{Then, we define recursive function $f$ where $f_i$ will be related to the maximum size of a cut that has exactly $s_j$ cut vertices in column $j$ (for all $j$) in the already processed part $G_i$.
More precisely, we want the recursive function $f$ to satisfy the properties later covered by Lemma~\ref{l_maxcut_unit}.
Once, $f$ satisfies the desired properties, we easily obtain the size of a maximum cut in the heavy part (Theorem~\ref{t_maxcut_sub}, below).

Now, we present a key observation for the construction of $f$.
Observe, by the properties of \U-bubble model, that the edges of $G_i$ can be partitioned into following disjoint sets:
\begin{compactenum}[]
\item $E_1=\{$edges of the graph $G_{i-1}\}$, 
\item $E_2=\{$edges inside $B_i\}$,
\item $E_3=\{$edges between $B_i$ and the same column above $B_i\}$,
\item $E_4=\{$edges between $B_i$ and the next column above $B_i\}$,
\item $E_5=\{$edges between $B_i$ and the bubble in the previous column and the same row as $B_i\}$,
\item $E_6=\{$edges between $B_i$ and column $C_0$ below $B_i\}$,
\item $E_7=\{$edges between $B_i$ and the bubble in column $C_{l+1}$ in the same row as $B_i\}$.
\end{compactenum}
Therefore, the idea there is to count the size of a desired cut of $G_i$ using the sizes of possible cuts in $G_{i-1}$ and add the size of a cut using edges $E_2-E_7$. The former is stored in $f_{i-1}$ and the later can be counted from the number of cut vertices in currently processed bubble $B_i$ and numbers in the $(l+1)$-tuple we are processing.
}

\if\povolenyobrazky1
\begin{figure}[bt]    
       \begin{center}
        \includegraphics{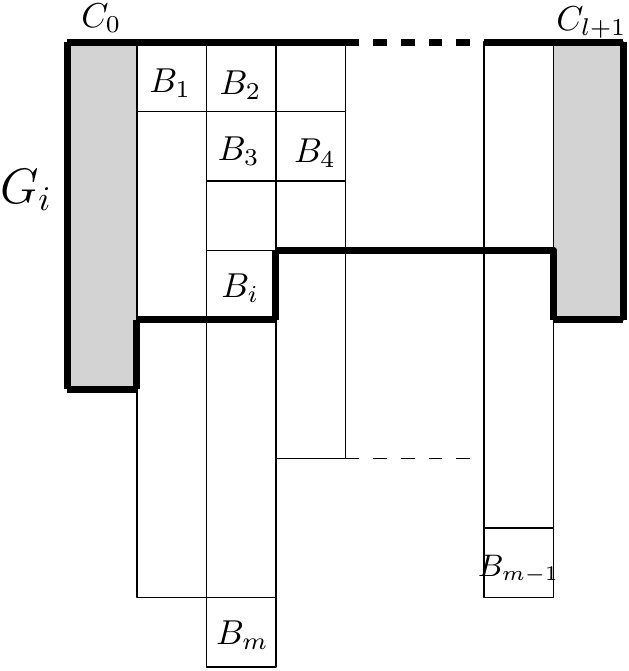}
        \end{center}
        \caption{A heavy part with light columns $C_0$ and $C_{l+1}$ and the highlighted subgraph $G_i$.}\label{f_Gi}
    \end{figure}
\fi

\toappendix{
Now, let us properly define the function $f$ and prove Theorem~\ref{thm_maxcut} formally.
We develop more notation. 
\lv{Let  $\hat\B$ be a heavy part with $l\ge 1$ heavy columns (numbered by $1,\dots,l$) and borders $C_0$ and $C_{l+1}$.}
Let $B_1$,\dots,$B_m$ be bubbles in $\hat\B\setminus(C_0\cup C_{l+1})$ numbered in the top-bottom, left-right order. Let $S_0$ and $S_{l+1}$ be (fixed) cuts in $C_0$ and $C_{l+1}$. 
To handle borders, we define auxiliary functions $n^{\downarrow}, n^{\leftarrow}, n^{\uparrow}$, $n^{\rightarrow}$ which output the number of cut vertices in borders in a specific position depending on the given row and column; they output $0$ if the given column  is not next to the borders. We define: 
\begin{compactenum}[$\bullet$]
\item 
 the number of (fixed) cut vertices in $C_0$ under the row $r$ (or 0 if the previous column is not~$C_0$):
 $$n^{\downarrow}(r,c)\df 
 \begin{cases}
 |S_0\cap\bigcup_{k=r+1}^{r_0}{B_{k,0}}| & c=1\\
 0 & c\neq 1,  
 \end{cases}\\
 $$
 \item  
 the number of (fixed) cut vertices of type $(*,+)$ in the left border $C_0$ in the row $r$:
$$n^{\leftarrow}(r,c)\df
\begin{cases}
|S_0\cap B^{*,+}_{r,0}| & c=1\\
 0 & c\neq 1,
\end{cases}$$
 \item 
 the number of (fixed) cut vertices in the right border $C_{l+1}$ above the row $r$:
 $$n^{\uparrow}(r,c)\df 
\begin{cases}
 |S_{l+1}\cap\bigcup_{k=1}^{r-1}{B_{k,l+1}}| & c=l\\
 0 & c\neq l,  
 \end{cases}\\
 $$
\item
 the number of (fixed) cut vertices of type $(+,*)$ in the right border $C_{l+1}$ in the row $r$: 
$$n^{\rightarrow}(r,c)\df 
\begin{cases}
 |S_{l+1}\cap B_{r,l+1}^{+,*}| & c=l\\
 0 & c\neq l. 
 \end{cases}\\
 $$
\end{compactenum}
We denote the number of vertices in $B_i$ by $b_i \df |B_i|$, analogously $b_i^{xy}\df |B_i^{xy}|$, $x,y\in\{+,-\}.$
We further denote the set of counts corresponding to all possible choices of cut vertices in the bubble  $B_i$  by $\beta_i$, i.e., 
   \begin{align*}
         \beta_i \df \big\{  (n_1,n_2,n_3,n_4) \mid 
                 \, &n_1 \in \{0,\dots,\bcc{i}\}, 
                 n_2 \in \{0,\dots,\bco{i}\}, 
                 n_3 \in \{0,\dots,\boc{i}\},\\ 
                 &n_4 \in \{0,\dots,\boo{i}\}, 
                 n_1+n_2+n_3+n_4 \le s_{\col{B_i}}\big\}.
    \end{align*} 
In addition, we denote the set of $(l+1)$-tuples characterizing all possible counts of cut-vertices in the $l$ heavy columns and an auxiliary number characterizing the count of possible edges from the last processed bubble, by 
\begin{align*}
T=\big\{(s_1,\dots,s_l, a) \mid\  
& a\in\N,\,\xspace 0\le a\le \max_{i\in\{1,\dots,m\}}{(\boc{i}+\bcc{i})},\,\\ 
& \forall j\in \{0,1,\dots,l\}: s_j\in\N,\,\xspace 0\le s_j \le c_j
\big\}.
\end{align*}
Let 
 $e(s_1,s_2)$ denote the number of cut-edges between two sets $S_1$, $S_2$ which are complete to each other and $S_k$, $k\in\{1,2\}$, contains $s_k$ cut vertices and $\ol{s_k}$ non-cut vertices,  i.e., $e(s_1,s_2) = s_1\cdot \ol{s_2}+\ol{s_1}\cdot s_2.$ We remark that it is important to know the numbers of non-cut vertices ($\ol{s_1}$ and $\ol{s_2}$), however, we will not write them explicitly for the easier formulas. It will be seen that they can be, for instance, stored in parallel with the numbers of cut vertices (or counted in each step again).

Finally, we define a recursive function $f$ by the following recurrence relation:
\begin{align*}
\hspace{-15pt}\forall (s_1,\dots,s_l, a)&\in T:\\
\text{if } s_1 \le b_{1},&\, s_2=\dots=s_l=0:\\
f_1((s_1,&\dots,s_l, a))=
    \max_{\substack{
        (\bcc{},\boc{},\bco{},\boo{})\in\beta_1:\\
        \bcc{}+\boc{}+\bco{}+\boo{}=s_1,\\
        \bcc{}+\boc{} = a
        }}
    { 
        \begin{aligned}[t] 
        \bigg(& e\left(s_1, n^\downarrow(1,1)\right)  + 
       s_1\cdot(b_1-s_1) \\
       &+ e\left(n^\leftarrow(1,1),(\bcc{}+\bco{})\right) \\
       &+ e\left(n^\rightarrow(1,1), \bcc{}+\boc{}\right) \bigg),
       \end{aligned}
        }\\
    \text{otherwise}&\text{: }\\
f_1((s_1,&\dots,s_l, a))= -\infty.
\end{align*}
\begin{align*}
\forall i\in&\{2,\dots,m\}, %
\forall (s_1,\dots,s_l, a)\in T:\\
&f_i((s_1,\dots,s_l,a))=    
 \max_{
    \substack{
         (\bcc{},\,\bco{},\,\boc{},\,\boo{}) \in \beta_i,\,
         z\in\N:\\
         \bcc{} + \boc{} = a,\\
         (s_1,\dots,s_c-b,\dots,s_l,z)\in T,\\
         z\le |B_{i-1}^{*+}|\\
     }}
     \begin{aligned}[t]
    \bigg(&  f_{i-1}((s_1,\dots,s_c-b,\dots,s_l,z))   \\ 
 & +  b\cdot (b_i-b)    \\ 
 & + e(b,s_{c+1})  + 
  e(b,s_c-b)  \\
    & + 
    e\left( n^\downarrow(r,c), b \right)  \\
   & + e\left( n^\leftarrow(r,c),(\bcc{}+\bco{}) \right)  \\ 
   & + e\left( (\bcc{}+\boc{}), n^\rightarrow(r,c) \right) 
     + A
     \bigg) 
    \end{aligned}%
\end{align*}
\begin{align*}
\vspace{-10pt}
\text{ w}&\text{here }A = 
 \begin{cases}
e(z\, , \bcc{}+\bco{})  & i>1,\, c=\col{B_{i-1}}+1,\\
 0 & \text{otherwise,}
 \end{cases}\\
\text{a}&\text{nd } b = \bcc{}+\bco{}+\boc{}+\boo{}, 
c = \col{B_i}, r= \row{B_i}, 
\text{ and } 
s_{l+1}= n^{\uparrow}(r,c).
\end{align*}
}

\toappendix{
\sv{\begin{proof}[Proof of Theorem~\ref{t_maxcut_sub}]}
We denote by $G_i$ 
the induced subgraph of $G(\hat\B)$ with 
$V(G_i)=B_1\cup\cdots\cup B_i\cup C_0 \cup C_{l+1}$ 
where 
$C_0$ and $C_{l+1}$ are borders of $\hat\B$.

\begin{lemma}\label{l_maxcut_unit}
For each $s=(s_1,\dots,s_l, a) \in T$ and for every $i \in \{1,\dots,m\}$, 
the value $f_i(s)$ is equal to the maximum size of a cut $S$ in $G_i$ that satisfies the following
\vspace{4pt}
\begin{compactenum}[$\bullet$]
\item for every $j\in \{1,\dots,l\}$, the number of cut vertices in the column $j$ in $G_i$ is equal to $s_j$, and $S$ agrees with $S_0\cup S_{l+1}$ in $C_0\cup C_{l+1}$, and

\item $a$ is equal to the number of cut vertices from $\Bcc{i}\cup\Boc{i}$,
\end{compactenum}
\vspace{4pt}
or $f_i$ is equal to $-\infty$ if there is no such cut.
\end{lemma}

\begin{proof}
We prove Lemma~\ref{l_maxcut_unit} by induction on the number of steps (bubbles).
Since $B_1$ is in the first heavy column, Lemma~\ref{l_maxcut_unit} is true for $i=1$ by Definition~\ref{d_u_bubblemodel} (iv).

In the inductive step, suppose that for every $s=(v_1,v_2,\cdots,v_l, z)\in T$, $f_{i-1}(s)$ is equal to the size of a maximum cut $S_{i-1}$ in $G_{i-1}$ such that the number of cut vertices in each column $j$, for every $j\in\{1,2,\dots,l\}$, in $G_{i-1}$ is equal to $v_j$, and the number of cut vertices from $B_{i-1}^{*+}$ is equal to $z$. Or $f_{i-1}(s)$ is equal to  $-\infty$ if such a cut does not exist.

As it was mentioned, the edges of $G_i$ can be partitioned into disjoint sets $E_1$---$E_7$. Recall:
\vspace{3pt} 
\begin{compactenum}[]
\item $E_1=\{$edges of the graph $G_{i-1}\}$, 
\item $E_2=\{$edges inside $B_i\}$,
\item $E_3=\{$edges between $B_i$ and the same column above $B_i\}$,
\item $E_4=\{$edges between $B_i$ and the next column above $B_i\}$,
\item $E_5=\{$edges between $B_i$ and the bubble in the previous column and the same row as $B_i\}$,
\item $E_6=\{$edges between $B_i$ and column $C_0$ below $B_i\}$,
\item $E_7=\{$edges between $B_i$ and the bubble in column $C_{l+1}$ in the same row as $B_i\}$.
\end{compactenum}
\vspace{3pt} 
Note that $E_6$ is non-empty only if $B_i$ is in the column $1$, similarly $E_7$ is non-empty only if $B_i$ is in the column $l$. 
Let $s=(s_1,\dots,s_l,a)\in T$ be fixed.
At first assume, $S$ is a maximum cut in $G_i$ (that agrees with $S_0\cup S_{l+1}$ in $C_0\cup C_{l+1}$) 
such that it contains $s_j$ vertices from the column $j$ for each $j\in\{1,2,\cdots,l\}$ and $a$ vertices from $B_i^{*+}$; we say $S$ \emph{satisfies the conditions} $s$. We discuss the case where no such cut exists, later.
We denote by $s^{xy}$ the number of vertices in $B_i^{xy}\cap S$, $x,y\in\{+,-\}$, and by $s'$ the sum of these values, i.e., $s'=s^{++}+s^{+-}+s^{-+}+s^{--}$. We denote $\col{B_i}$ by $j$, and $\row{B_i}$ by $r$. Then,
$$
E(S,\ol{S}) = (E(S,\ol{S})\cap E(G_{i-1}) ) 
\cup \left\{uv\in E_k \mid u\in S,v\notin S, k\in\{2,\dots,6\}\right\}.
$$
Which leads to the expression:
\begin{align*}%
|E(S,\ol{S})| =&  |E(S,\ol{S})\cap E(G_{i-1})|   \\
&+ s' \cdot (b_i-s')\\
&+ e(s', (s_j-s')) \\
& + e(s', s_{j+1}) \\
& + A  \\
&+ e(s', n^{\downarrow}(r,j)) + e(s^{++}+s^{+-}, n^{\leftarrow}(r,j))  \\
&+ e(s^{++}+s^{-+}, n^{\rightarrow}(r,j)),\\
\text{ where }A =& \begin{cases}
e(|S\cap B_{i-1}^{*+}|, s^{++}+s^{+-}) & j=\col{B_{i-1}}+1,\\
0 & otherwise.
\end{cases}
\end{align*}
By the induction hypothesis, 
 \begin{align}
  |E(S,\ol{S})\cap E(G_{i-1})|  \le f_{i-1}(s_1,\cdots,s_j-s',\cdots,s_l, |S\cap B_{i-1}^{*+}|). 
\label{a_cut}
  \end{align}
It gives us together with the right part of the equation, %
 the definition of $f_i$ for $b^{xy}=s^{xy}$, $b=s'$ and $z=|S\cap B_{i-1}^{*+}|$. 
 Therefore, $$|E(S,\ol{S})|\le f_i(s).$$
Furthermore, we show that $f_i(s)$ is the size of a cut satisfying the conditions $s$.
Since the value of the function $f_{i-1}((s_1,\cdots,s_j-b,\cdots,s_l, z))$ is for any number $b\in\{0,\dots,\min{(s_j,b_i)}\}$ a size of a cut in $G_{i-1}$ which satisfies the conditions $(s_1,\cdots,s_j-b,\cdots,s_l, z)$, or $-\infty$ (if no such cut exists), we can extend that cut into $G_i$ by adding $b^{xy}$ vertices from $B_i^{xy}$ where $\bcc{}+\boc{}=a$ and $\bcc{}+\bco{}+\boc{}+\boo{}=b$. 
Consequently, $f_i(s)$ is a size of a cut on $G_i$ satisfying that it contains $s_i$ vertices from the column $i$ and $a$ vertices from $|B_{i}^{*+}|$. 
At least one such cut exists, by (\ref{a_cut}). 
Therefore, $|E(S,\ol{S})|\ge f_i(s)$. It leads to the equation $|E(S,\ol{S})|= f_i(s)$, otherwise, $S$ is not a maximum cut.

In a similar way, we can extend every cut on $G_{i-1}$ to $G_i$. Therefore, if there exist no cut on $G_i$ which satisfies the conditions $s$, there exists no cut in $G_{i-1}$ which can be extended to the cut on $G_i$ satisfying the conditions $s$. Consequently, $f_i(s)=-\infty$ by the definition of $f$ since $f_{i-1}(v)=-\infty$ for all $(l+1)$-tuples $v$ which appear in the definition.
\end{proof}
}

Finally, we obtain the next theorem about a maximum cut of a heavy part as a corollary of Lemma~\ref{l_maxcut_unit}.

\begin{thm}\label{t_maxcut_sub}
Let $\hat\B$ be a heavy part with $l\ge 1$ heavy columns (numbered by $1,\dots,l$) and borders $C_0$ and $C_{l+1}$. Let $B_1$,\dots,$B_m$ be bubbles in $\hat\B\setminus(C_0\cup C_{l+1})$ numbered in the top-bottom, left-right order. Let $S_0$ and $S_{l+1}$ be (fixed) cuts in $C_0$ and $C_{l+1}$. Then, the size of a maximum cut in $G(\hat\B)$ that agrees with $S_0\cup S_{l+1}$ in light columns is
\lv{$$ \mcs{G(\hat\B), S_0\cup S_{l+1}} = \max_{s\in T} f_m(s).$$ }
\sv{$\mcs{G(\hat\B),S_0\cup S_{l+1}} = \max_{s\in T} f_m(s).$ }
\end{thm}

Towards proving Theorem~\ref{thm_maxcut} and Corollary~\ref{c_columns}, it remains to prove the time complexity of processing a heavy part.

\begin{lemma}\label{l_maxcut_heavy}

Let $\hat\B$ be a heavy part with $l\ge 1$ columns, $m$ bubbles, and a fixed cut in the borders.
The size of a maximum cut of $\hat\B$ that agrees with the fixed cut in the borders can be determined in the time:$$(c_1+1)\cdots(c_l+1)\cdot (a+1) \cdot \sum_{i=1}^m{\left(\bcc{i}\cdot \bco{i}\cdot \boc{i}\cdot\boo{i}\right)}$$
where 
$c_j$ is the number of vertices in the column $j$, i.e., $c_j = \sum_{i'=1}^{r_j}B_{i',j}$, and $a = \max_{i}{|B_i^{*+}|}$.
\end{lemma}
\begin{proof}
We analyze the time complexity of the algorithm from Lemma~\ref{l_maxcut_heavy}. 
Let $T$ denote all the possible $(l+1)$-tuples. 
Then 
$|T| = (c_1+1)\cdots(c_l+1)\cdot (a+1).$
The time for processing a bubble $B_i$ is $|T|\cdot \bcc{i}\cdot \bco{i}\cdot \boc{i}\cdot\boo{i}$.
The time complexity of processing $\hat\B$ is then 
\[|T|\cdot \sum_{i=1}^m{\bcc{i}\cdot \bco{i}\cdot \boc{i}\cdot\boo{i}}.\qedhere\]%
\end{proof}

Now, we are ready to prove Theorem~\ref{thm_maxcut}.

\begin{proof}[Proof of Theorem~\ref{thm_maxcut}]
By Lemma~\ref{l_maxcut_divideAndConquer}, heavy parts can be processed independently on each other, given a cut on their borders. Moreover, it is sufficient for a light column $C$ to remember only the biggest cuts on the left of $C$ (containing $C$) for each possible cut in $C$. Therefore, there is no need to guess cuts in all light columns at once. %
It is sufficient to guess a cut only in two consecutive light columns at once.

Observe that there are at most $2^{\sqrt{n}}$ guesses of cut vertices for a light column and there are at most $n$ light columns. 
Therefore, the time complexity of determining the size of a maximum cut in $G$ is at most 
$n\cdot (2^{\sqrt{n}})^{2} \cdot P,$ where $P$ is the maximum time for processing a heavy part.
Now, we want to prove that a time complexity of processing a heavy part $\hat \B = \langle B_{i,j}\rangle_{1\le j\le l, 1\le i\le r_j}$ with a given guess of light columns is $2^{\tilde{{\cal O}}(\sqrt n)}.$  

By Lemma~\ref{l_maxcut_heavy}, the time complexity of processing a heavy part with a given guess of light columns is  
\begin{align*}
P&=(c_1+1)\cdots(c_l+1)\cdot (a+1)\cdot \sum_{i=1}^m{\left(\bcc{i}\cdot \bco{i}\cdot \boc{i}\cdot\boo{i}\right)}\\
 &\le (n+1)^{1+\sqrt{n}}\cdot \sum_{i=1}^m{b_i^4}
 \le (n+1)^{1+\sqrt{n}} \cdot n^4 \in 2^{\tilde{{\cal O}}(\sqrt{n}}).
 \end{align*}

To sum up, we can determine the size of a maximum cut in the time: $n\cdot (2^{\sqrt{n}})^{2} \cdot P\in 2^{\tilde{{\cal O}}(\sqrt{n})}.$ 
For brevity, we analyzed only the size of a maximum cut. However, the maximum cut itself can be determined retroactively in the time $2^{\tilde{{\cal{O}}}(\sqrt{n})}$, as well.
\end{proof}
\lv{%

}%
\lv{Lemma~\ref{l_maxcut_heavy} has a nice corollary for graphs with a \U-bubble model with a constant number of columns. According to Lemma~\ref{l_maxcut_heavy}, we are able to solve the \maxcut problem in those graphs in polynomial time which is formulated as Corollary~\ref{c_columns} in the introduction.}%
  \lv{%
Therefore, we improved another polynomial-time algorithm by Boyaci, Ekim, Shalom~\cite{BoyaciES18} solving the \maxcut problem in co-bipartite chain graphs with possible twins (which is exactly the class of graphs defined by a classic bubble model with only two columns).}

\toappendix{%
\lv{\begin{proof}}
\sv{\begin{proof}[Proof of Corollary~\ref{c_columns}]}
Let $G$ be a graph on $n$ vertices which is defined by a \U-bubble model $\B$ with $k$ columns and $m$ bubbles.
The bubble model $\B$ can be seen as a heavy part with no cut-vertices in its borders. 
By Lemma~\ref{l_maxcut_heavy}, the size of a maximum cut in \B can be determined in time
$T= (c_1+1)\cdots(c_k+1)\cdot (a+1) \cdot \sum_{i=1}^m{\left(\bcc{i}\cdot \bco{i}\cdot \boc{i}\cdot\boo{i}\right)}$
where $b_i^{xy}, xy\in\{+,-\}$ is the number of vertices in the bubble quadrant $B_i^{xy},$ 
and $c_j$ is the number of vertices in the column $j$, i.e., $c_j = \sum_{i'=1}^{r_j}B_{i',j}$, and $a = \max_{i}{|B_i^{*+}|}$.

By  Arithmetic Mean-Geometric Mean Inequality (AM-GM) we obtain
\begin{align*}
T&\le(a+1)\cdot\left(\frac{1}{k}\cdot\sum_{j=1}^{k}(c_j+1)\right)^k \cdot \sum_{i=1}^m{\left(\frac{\bcc{i} + \bco{i} + \boc{i} + \boo{i}}{4}\right)^4}\\
&=(a+1)\cdot \left(\frac{n+k}{k} \right)^k \cdot \sum_{i=1}^{m}\left(\frac{{b_i}}{4} \right)^4\\
&\le(a+1)\cdot \left(\frac{n+k}{k} \right)^k \cdot \left(\frac{n}{4} \right)^4 
\in {\cal O}(n^{k+5}).
\end{align*}
It remains to prove the special case where $k=2.$ Notice, it is sufficient to distinguish only between vertices in quadrants of types $(*,+)$ and $(*,-)$ in the first column, and simillarly $(+,*)$ and $(-,*)$ in the second column. Therefore, we obtain $\left(\frac{{b_i}}{2} \right)^2$ instead of $\left(\frac{{b_i}}{4} \right)^4$ which leads to the time complexity ${\cal O}(n^{k+1+2})={\cal O}(n^5).$
\end{proof}

Note that Theorem~\ref{t_maxcut_sub} states the explicit size of a maximum cut.
}

\section{Clique-width of mixed unit interval graphs}\label{s_cwd}

The clique-width is one of the parameters which are used to measure the complexity of a graph.   
\sv{%
  \toappendix{\section{Additions to Section~\ref{s_cwd}}\label{a_cwd}}%
}%
\toappendix{%
Many \NP-hard problems, those which are expressible in Monadic Second-Order Logic using second-order quantifiers on vertices ($MSO_1$), 
can be solved efficiently in graphs of bounded clique-with~\cite{CourcelleMR00}. For instance, 3-coloring.
}%
Definition of the clique-width is quite technical but it follows the idea that a graph of the clique-width at most $k$ can be iteratively constructed such that in any time, there are at most $k$ types of vertices, and vertices of the same type behave indistinguishably from the perspective of the newly added vertices.
\sv{Formally, the clique-width of a graph $G$, denoted by $\cwd{G}$, is the smallest integer number of different labels that is needed to construct the graph $G$ using the four operations: 
creation of a labeled vertex, 
disjoint union of labeled graphs, 
renaming all labels $i$ to $j$, and 
connecting all vertices with label $i$ to all vertices with label~$j$, $i\neq j$ (already existing edges are not doubled).
 Such a construction of a graph can be represented by an algebraic term composed of the operations, called \emph{cwd-expression}.}%
\sv{%

We present here the main result for the better upper-bounds on clique-width 
which is inspired by a similar result for unit interval graphs~\cite{HeggernesMP09}.
In general, unit interval graphs (and therefore mixed unit interval graphs) have unbounded clique-width~\cite{Golumbic00}
and the known upper-bound (even for interval graphs) is the size of a maximum clique + 1~\cite{KaplanS96, FellowsRRS06}.
The proofs can be found in the full version.
}

  \begin{definition}[Courcelle 2000]\label{def_cwd}
The clique-width of a graph $G$, denoted by $\cwd{G}$, is the smallest integer number of different labels that is needed to construct the graph $G$ %
 using the following operations:
\begin{compactenum}
\item[0.] creation of a vertex with label $i$,
\item[1.] disjoint union (denoted by $\oplus$),
\item[2.] relabeling: renaming all labels $i$ to $j$ (denoted by $\rho_{i\rightarrow j}$),
\item[3.] edge insertion: connecting all vertices with label $i$ to all vertices with label~$j$, $j\in\{1,\ldots,k\},$ $i\neq j$; already existing edges are not doubled (denoted by $\eta_{i,j}$).
\end{compactenum}
\end{definition}
 Such a construction of a graph can be represented by an algebraic term composed of the operations $\oplus$, $\rho_{i\rightarrow j}$, and $\eta_{i,j}$, called \emph{cwd-expression}. We call \emph{k-expression} a cwd-expression in which at most $k$ different labels occur. Using this, we can say that the clique-width of a graph $G$ is the smallest integer $k$ such that the graph $G$ can be defined by a $k$-expression.

\begin{example}\label{ex_cwd_complete}
The diamond graph $G$ on the four vertices $u, v, w, x$ (the complete graph $K_4$ without an edge $vw$) is defined by the following cwd-expression:
$$
\edge{1}{2}(\relabel{2}{1}(\edge{1}{2}( 1(u) \oplus 2(v) \oplus 2(w)))\oplus 2(x)).
$$
Therefore, $\cwd{G}\le 2$.
\end{example}

Fellows, Rosamond, Rotics, and Szeider \cite{Fellows09} proved in 2009 that the deciding whether the clique-width of a graph $G$ is at most $k$ is \NP-complete.
Therefore, researchers put effort into computing an upper-bound of the clique-width.

Courcelle and Olariu~\cite{Courcelle00} showed in 2000 that bounded treewidth implies bounded clique-width (but not vice versa). They showed that for any graph $G$ with the treewidth $k$, the clique-width of $G$ is at most $4 \cdot 2^{k-1} + 1$.

Golumbic and Rotics \cite{Golumbic00} proved that unit interval graphs have unbounded clique-width, consequently, (mixed unit) interval graphs have unbounded clique-width as well. Therefore, computing upper-bounds are of particular interest. Fellows et al.~\cite{FellowsRRS06} showed that the clique-width of a graph is bounded by its pathwidth~+~2, therefore, the clique-width of interval graphs as well as of unit interval graphs is upper-bounded by the size of their maximum clique + 1~\cite{KaplanS96, FellowsRRS06}. 
Using a bubble model structure, subclasses of unit interval graphs were characterized in terms of (linear) clique-width~\cite{Lozin2011, Meister2015}.
Courcelle~\cite{Courcelle00} observed that clique-width can be computed componentwise.

\begin{lemma}[Courcelle 2000,\cite{Courcelle00}]
Any graph $G$ satisfies that 
$$\cwd{G} = max{\{\cwd{G'}  \mid G' \text{ is a connected component of } G\}}.$$
\end{lemma}

We provide an upper-bound of the clique-width of a graph $G$ depending on the number of columns in a \U-bubble model of $G$. We express it also in the size of a maximum independent set.  

\begin{lemma} \label{l_cwd_col}
Let $G$ be a mixed unit interval graph and \B be a \U-bubble model of $G$. Then $\cwd{G}\le k + 3$, where $k$ is the number of columns of \B. Moreover, a $(k+3)$-expression defining the graph $G$ can be constructed in $\mathcal{O}(n)$ time from \B. 
\end{lemma}
\begin{proof}
The proof is inspired by the proof for unit interval graphs~\cite{HeggernesMP09}.

We find a $(k+3)$-expression defining $G$ and, therefore, prove that $\cwd{G}\le k + 3$.
We use $k+3$ labels where label $i$ will be assigned to $i$-th column of \B and the remaining three labels, denoted by $l_1, l_2, l_3$, are used for maintaining the last two added vertices. 

We define a linear order on vertices of $G$ according to \B as follows:
\vspace{4pt}
\begin{compactenum}
\item[(i)] We take the vertices from top to bottom, left to right. Formally, let 
$x\in B_{i,j},$ $y\in B_{l,m}$, we define $x\prec y$ if $i<l$ or $(i=l \text{ and } j<m);$
\item[(ii)] we define the following order on  bubble quadrants: 
    $$x\prec y\prec z\prec w\text{ for } x\in B_{i,j}^{--}, y\in B_{i,j}^{+-}, z\in B_{i,j}^{-+}, w\in B_{i,j}^{++};
   $$ 
\item[(iii)] we define an arbitrary linear order on vertices in the same quadrant of the same bubble.
\end{compactenum}
\vspace{4pt}

The idea of the proof is that every column has its own label and we need three more labels for maintaining the last added vertices. We will add vertices to $G$ in the described order which ensures that a new vertex is complete to all vertices from the following column and anti-complete to all vertices from the previous column except those from the same row. Recall that according to the definition of \mixed-bubble model, there is an edge between vertices $x\in B_{i,j}$ and $y\in B_{i,j+1}$ if and only if $x\in B_{i,j}^{*,+}$ and $y\in B_{i,j+1}^{+*}$. Therefore, vertices from the last constructed bubble in the previous column must have two distinct labels according to the types of the vertices. 
 However, once we add all vertices from the actual bubble, we do not need to distinguish between vertices from the previous column, anymore. Therefore, we rename their labels to the label of their column. 

Formally. Let $x$ be the first (smallest) vertex of $G$ according to the defined linear order. We know that $x$ is from the first column by Definition~\ref{d_u_bubblemodel} (iv). If $x$ is of type $(-,+)$ or $(+,+)$, we label it by $l_1$, otherwise by 1, so the expression for $G[\{x\}]$ is $1(x)$ if $x$ is of type $(+,-)$ or $(-,-)$, and $l_1(x)$ otherwise.

Let $y$ be the first non-processed vertex from $G$, i.e., a label is assigned to all preceding vertices.
Let %
$l_2, l_3 \in \{k+1, k+2, k+3\}$ are currently unused labels or $l_2$ is used in the actual bubble $B_{i,j}$ and $l_3$ is unused, and $l_1$ may be used (in the previous column). Note that at most one label from $\{k+1,k+2,k+3\}$ is used in the previous column any time.
We split the proof according to the type of $y$, the bubble quadrant where $y$ belongs.

\begin{itemize}
\item[(a)] $y\in B_{i,j}^{--}$. We use label $l_2$ for $y$. Then, we make $y$ (the only one vertex with label $l_2$) complete to vertices with labels $j+1$ (if $j<k$) and $j$. Relabel $l_2$ to $j$.
\item[(b)] $y\in B_{i,j}^{+-}$. We use label $l_2$ for $y$. Then, we make $y$ (the only one vertex with label $l_2$) complete to vertices with labels $j+1$ (if $j<k$), $j$, $l_1$. Relabel $l_2$ to $j$.
\item[(c)] $y\in B_{i,j}^{-+}$. We use label $l_2$ for $y$. Then, we make all vertices with label $l_2$ complete to vertices with labels $j+1$ (if $j<k$), $j$, $l_2$. (Do not relabel vertices with label $l_2$).

\item[(d)] $y\in B_{i,j}^{++}$. We use label $l_3$ for $y$. Then, we make  $y$ (the only one vertex with label $l_3$) complete to vertices with labels $j+1$ (if $j<k$), $j$, $l_1$, $l_2$.  Relabel $l_3$ to $l_2$.

\end{itemize}
If all vertices from $B_{i,j}$ were used, we rename all vertices with the label $l_1$ to $j-1$ if $j>1.$ If $j=k$, we relabel $l_2$ to $k$.

For the correctness, observe that the previous column has always at most two labels and in a), b), and d) the temporary label for $y$ is unique (no other vertices are labeled by it at that time).    
The rest follows from the definition of adjacency in the \mixed-bubble model.
Since we constructed $G$ using at most $k+3$ labels, $\cwd{G}\le k+3$.

The described algorithm processes each vertex once and each vertex has at most three labels in total. Moreover, the algorithm needs a constant work for each vertex---for instance, a cwd-expression for the option a) is:
$$ \rho_{l_2\rightarrow j}(\eta_{j,l_2}(\eta_{j+1,l_2}(l_2(y)\oplus G'))),$$
where $G'$ is the already constructed graph before adding the vertex $y$.
Therefore, the $(k+3)$-expression defining $G$ is constructed in linear time given a \U-bubble model in an appropriate structure.
\end{proof}

\begin{thm}\label{thm_cwd_mis}
Let $G$ be a mixed unit interval graph. Then $\cwd{G}\le 2\Mis{G} + 3$. Moreover, a $(2\Mis{G}+3)$-expression defining the graph $G$ can be constructed in $\mathcal{O}(n)$ time provided a \U-bubble model of $G$ is given.
\end{thm}
\begin{proof}
We apply Lemma~\ref{l_cwd_col} and Lemma~\ref{l_bubble_mis_columns} together to obtain the statements. 
\end{proof}

Next, we provide a different bound for clique-width which is obtained by a small extension of the proof for unit interval graphs using the bubble model by Heggernes, Meister, and Papadopoulos \cite{HeggernesMP09}. We include the full proof for completeness. 

\sv{To state the main theorem, we}\lv{We} need more notation. Let $G$ be a mixed unit interval graph and let $\B = \langle B_{i,j}\rangle_{1\le j\le k, 1\le i\le r_j}$ be a \U-bubble model for $G$. We say that vertices from the same column $j$ of~$\B$ create a \emph{group} if they have the same neighbours in the following column $j+1$ of~\B. Let $v\in \Bij$, the \emph{group number of vertex $v$} in \B, denoted by $g_\B(v)$, is defined as the maximum number of groups in 
$N(v)\cap\big(
\bigcup_{i'=i+1}^{r_{j-1}} B_{i',j-1} \cup \bigcup_{i'=1}^{i-1} B_{i',j} 
\cup A
\big)$
 over the sets
 $A= B_{i,j-1}^{*+}\cup B_{i,j}^{+*}$ and $A= B_{i,j}$.
Then the \emph{group number of $G$} in \B is defined as 
\sv{$\varphi_\B(G) \df \max_{v\in V(G)}{g_\B(v)}$.}%
\lv{\[\varphi_\B(G) \df \max_{v\in V(G)}{g_\B(v)}.\]}%

\toappendix{%
\begin{lemma}\label{l_cwd_phi_clique}
Let $G$ is a mixed unit interval graph and \B a \U-bubble model for $G$. The following inequality holds
\[
\varphi_\B(G) \le \Mclique{G} - 1. 
\]
\end{lemma}
\begin{proof}
Let $v\in \Bij.$ 
Observe that 
$\bigcup_{i'=i+1}^{r_{j-1}} B_{i',j-1} \cup \bigcup_{i'=1}^{i-1} B_{i',j} 
\cup A \cup \{v\}$ is a clique (for both the possibilities of $A$), see Lemma~\ref{l_bubble_clique_columns}. Moreover, $v$ is not included in the counting the group number of $v$, and no vertex can be in more than one group. Therefore, 
$g_\B(v) \le \Mclique{G} - 1$ for any vertex $v$ which leads to the desired inequality. 
\end{proof}

\begin{thm} \label{thm_cwd_groups}
Let $G$ be a mixed unit interval graph and \B a \U-bubble model for~$G$. Then $\cwd{G}\le \varphi_\B(G)+2.$ Moreover, a $(\varphi_\B(G)+2)$-expression defining the graph $G$ can be constructed in $\mathcal{O}(n+m)$ 
    time provided a \U-bubble model of $G$ is given.
\end{thm}
\begin{proof}
Our aim is to find a $(\varphi_\B(G)+2)$-expression defining $G$. 
We add vertices in the order from left to right, top to bottom of \B processing vertices of type $(+,*)$ at first, i. e., in the following linear order:

\begin{itemize}
\item[(i)] $x\prec y$ for $x\in B_{i,j},$ $y\in B_{l,m}$, where $j<m$ or $(j=m \text{ and } i<l);$
\item[(ii)] $x\prec y\prec z\prec w\text{ for } x\in B_{i,j}^{++}, y\in B_{i,j}^{+-}, z\in B_{i,j}^{-+}, w\in B_{i,j}^{--};
   $ 
\item[(iii)] an arbitrary linear order on the vertices in the same quadrant of the same bubble.
\end{itemize}

Now, we follow the original proof. Shortly,  we add each vertex $v$ in a proper way. We assume that a label is assigned for each previous vertex and all the vertices which belong to the same group have the same label. At first, we change to 1 the label of all the previous vertices which are nonadjacent to $v$. We know that at most $g_\B(v)$ distinct labels are used in the remaining groups, say labels $\{2,\dots,g_\B(v)+1\}$. This is true since all the groups are adjacent to $v$ and because of the linear order. 

Note that it is important to add first all the vertices of type $(+,*)$ from a bubble. 
Otherwise, $g_\B(v)+1$ remaining groups could be there; in the situation that $v$ is of type $(+,*)$, a potentially one distinct label is needed for $B_{i,j-1}^{*+}$, and another for $B_{i,j}^{*-}$. One the other hand, if all the vertices of type $(+,*)$ precede vertices of type $(-,*)$ in one bubble, this situation does not happen---a potential label of $B_{i,j-1}^{*+}$ would be released. Therefore, it is enough to take into account only the parts $A= B_{i,j-1}^{*+}\cup B_{i,j}^{+*}$, and $A= B_{i,j}$, and not the bigger one $A= B_{i,j-1}^{*+}\cup B_{i,j}$, in the definition of $g_\B(v)$.

We use a free label, say $g_\B(v)+2$, for $v$ and join all the vertices with this label with vertices with labels $2,\dots,g_\B(v)+1$. Next, change the label of $v$ to a label of its group if $v$ belongs to an already existing group. 
We continue with the next vertex. During the processing of each vertex, we need no more than its group number~+~2 distinct labels. Therefore,  $\cwd{G}\le \varphi_\B(G)+2.$

It remains to determine the running time for the construction of the expression defining $G$. 
Assume a \U-bubble model is given in a way that going over all vertices takes linear time in the number of vertices. 
First, we count the time for the creation of groups. 
For each vertex $v$ we compare its neighbors from the next column with the neighbors of the previous vertex in this column. 
Therefore, the splitting vertices into groups and determining the group number of $G$ take $\mathcal{O}(m+n)$ time. 
In a constant time, we determine a free label for each vertex. Then, we need to check the labels of groups in the neighborhood of each vertex $v$ and create a $\mathcal{O}(g_\B(v))$ long cwd-expression, 
yielding $\mathcal{O}(m+n)$ time in total. Furthermore, each vertex is at most once relabeled to 1 since once it is relabeled to 1, its label remains 1 for the rest of the algorithm. Therefore, the relabeling of vertices that are nonadjacent to a newly added vertex takes $\mathcal{O}(n)$ time in total. To sum up, the algorithm outputs the construction in $\mathcal{O}(n+m)$ time.
\end{proof}
}

\lv{Combination of Theorems~\ref{thm_cwd_mis} and \ref{thm_cwd_groups} give us the following theorem.}

\begin{thm}\label{thm_main_cwd}
Let $G$ be a mixed unit interval graph. Then 
$$\cwd{G} \le \min{\{ 2\Mis{G}+3, \varphi_\B(G) + 2 \}} \le \Mclique{G} + 1,$$ 
where $\B$ is a \U-bubble model for $G$. Moreover, the corresponding expression can be constructed in $\mathcal{O}(n+m)$ time providing \B is given, otherwise in $\mathcal{O}(n^2)$ time. 
\end{thm}
Observe that $\varphi_\B(G)\le 2\max{\{r_j \mid 1\le j\le k\}}$. 
\lv{A combination of Theorem~\ref{thm_main_cwd} and Lemma~\ref{l_cwd_col} gives}%
\sv{We also obtain }%
a useful Corollary~\ref{cor_cwd_rows_columns}. In particular, if the number of rows or number of columns is bounded, than clique-width is bounded.

\begin{cor}\label{cor_cwd_rows_columns}
Let $G$ be a mixed unit interval graph. Then $\cwd{G} \le \min{\{ k+3, 2r+2}\},$
where $k$ is the number of columns and $r$ is the length of a longest column in a \U-bubble model for $G$. 
\end{cor}

Note that by an application of Lemma 4.1 in~\cite{Lozin2011}, slightly worse bounds on clique-width in terms of rows and columns can also be derived.
In particular, if we take two natural orderings of the bubbles in the \U-bubble  model, one taking rows first and the other taking columns first, we obtain two times larger multiplicative factor than in Corollary~\ref{cor_cwd_rows_columns}.

\section{Conclusion}\label{s_conclusion}
\sv{
\toappendix{\section{Additions to Section~\ref{s_conclusion}}\label{a_conclusion}}
}
\lv{
The main contribution of this work is a new representation of mixed unit interval graphs---the \U-bubble model. This structure is particularly useful in the design of algorithms and their analysis.
Using the \U-bubble model, we presented new upper-bounds for the clique-width of mixed unit interval graphs and designed a subexponential-time algorithm for the \maxcut problem on mixed unit interval graphs. We further realized that the state-of-the-art polynomial-time algorithm for the \maxcut problem on unit interval graphs is incorrect.
}
A long-term task is to determine the difference between the time complexity of basic problems on unit interval graphs compared to interval graphs.
In particular, on a more precise scale of mixed unit interval graphs, determine what is a key property for the change of the complexity.
Independently, a long-standing open problem is the time complexity of the MaxCut problem on unit interval graphs, in particular, decide if it is NP-hard or polynomial-time solvable.
An interesting direction to pursue the first task could be the study of labeling problems;
either $L_{2,1}$-labeling or Packing Coloring. 
\toappendix{%
\sv{%

}%
\lv{Both problems were motivated by assigning frequencies to transmitters.}%
\sv{Both, $L_{2,1}$-labeling or Packing Coloring, were motivated by assigning frequencies to transmitters.}%
The $L_{2,1}$-labeling problem was first introduced by Griggs and Yeh in 1992 \cite{GriggsY92}. The packing coloring problem is newer, it was introduced by Goddard et al.\ in 2008~\cite{GoddardHHHR08}.
}%
Although, these are well-known problems, quite surprisingly, their time complexity is open for interval graphs.
\lv{%

}%
\toappendix{%
The $L_{2,1}$-labeling problem assigns labels $\{0,\ldots,k\}$ to vertices such that the labels of neighboring vertices differ by at least two and the labels of vertices in distance two are different.
}%
\lv{The time complexity of this problem is still wide open even for unit interval graphs, despite partial progress on specific values for the largest used label. Sakai proved that the value of the largest label lies between $2\chi-2$ and $2\chi$ where $\chi$ is the chromatic number~\cite{Sakai94}.}%
\sv{The complexity of $L_{2,1}$-labeling is still wide open even for unit interval graphs, despite partial progress on specific values for the largest used label~\cite{Sakai94}.%
}%
\lv{%

}%
\toappendix{
The packing coloring problem asks for an existence of such a mapping $c:V\to\{1,\ldots,m\}$ that for all $u\neq v$ with $c(u)=c(v)=i$ the distance between $u$ and $v$ is at least $i$.
\lv{This problem is wide open on interval graphs.}
}
Recently, there was a small progress on unit interval graphs leading to an FPT algorithm (time $f(k)\cdot n^{{\cal O}(1)}$ for some computable function $f$ and parameter $k$).
It is shown in~\cite{KimLMP18} that the packing coloring problem is in FPT parameterized by the size of a maximum clique. We note that the algorithm can be straightforwardly extended to mixed unit interval graphs.
However, a polynomial-time algorithm or alternatively NP-hardness for (unit) interval graphs is of a much bigger interest.

\lv{
\bigskip
\paragraph{Acknowledgements}
The authors would like to thank Vít Jelínek for helpful comments.
This paper is based on the master thesis of Jana Novotná~\cite{DiplomaThesis}.
}

\lv{\bibliographystyle{plainurl}}
\sv{\clearpage}
\sv{\bibliographystyle{plainurl}}%
\bibliography{bibliography}

\clearpage
\appendix
\appendixText

\end{document}